\documentclass[preprint,12pt]{elsarticle}

\usepackage{amsmath,amsfonts,amsthm,amssymb,paralist,subfigure,graphicx,amsbsy,float,epsfig,color}
\usepackage{cuted,mathtools,lipsum}

\usepackage{algorithm,algpseudocode}

\usepackage{stmaryrd,url}

\usepackage{amsthm}
\newtheorem{theorem}{Theorem}
\newtheorem{lemma}{Lemma}
\newtheorem{ass}{Assumption}
\newtheorem{definition}{Definition}
\newtheorem{remark}{Remark}
\newtheorem{corol}{Corollary}

\def\ve{\varepsilon}

\def\mb{\mathbf}

\def\mc{\mathcal}

\DeclareMathOperator*{\argmin}{argmin}
\DeclareMathOperator*{\argmax}{argmax}

\journal{Systems and Control Letters}

\begin{document}

\begin{frontmatter}



\title{Distributed Delay-Tolerant Strategies for Equality-Constraint Sum-Preserving Resource Allocation}


\author[Sem]{Mohammadreza Doostmohammadian}
\affiliation[Sem]{Faculty of Mechanical Engineering, Semnan University, Semnan, Iran, doost@semnan.ac.ir.}
\author[AA]{ Alireza Aghasi} 
\affiliation[AA]{Department of Electrical Engineering and Computer Science, Oregon State University, USA, name.surname@oregonstate.edu.}
\author[MV]{Maria Vrakopoulou}
\affiliation[MV]{University of Melbourne, Melbourne, Australia, maria.vrakopoulou@unimelb.edu.au}
\author[HRR]{ Hamid R. Rabiee}
\affiliation[HRR]{Department of Computer Engineering, Sharif University of Technology, Tehran, Iran, rabiee@sharif.edu.}
\author[UK]{ Usman A. Khan}
\affiliation[UK]{Electrical and Computer Engineering Department, Tufts University, MA, USA, khan@ece.tufts.edu.} 
\author[Aalto,TC]{Themistoklis Charalambous} 
\affiliation[Aalto]{School of Electrical Engineering, Aalto University, Espoo, Finland, name.surname@aalto.fi.}         
\affiliation[TC]{School of Electrical Engineering, University of Cyprus, Nicosia, Cyprus, surname.name@ucy.ac.cy.}

\begin{abstract}
	This paper proposes two nonlinear dynamics to solve constrained distributed optimization problem for resource allocation over a multi-agent network. In this setup, coupling constraint refers to resource-demand balance which is preserved at all-times. The proposed solutions can address various model nonlinearities, for example, due to quantization and/or saturation. Further, it allows to reach faster convergence or to robustify the solution against impulsive noise or uncertainties. We prove convergence over weakly connected networks using convex analysis and Lyapunov theory. Our findings show that convergence can be reached for general sign-preserving odd nonlinearity. We further propose delay-tolerant mechanisms to handle general bounded heterogeneous time-varying delays over the communication network of agents while preserving all-time feasibility. This work finds application in CPU scheduling and coverage control among others. This paper advances the state-of-the-art by addressing \emph{(i)} possible nonlinearity on the agents/links, meanwhile handling \emph{(ii)} resource-demand feasibility at all times,  \emph{(iii)} uniform-connectivity instead of all-time connectivity, and \emph{(iv)} possible heterogeneous and time-varying delays. To our best knowledge, no existing work addresses contributions \emph{(i)}-\emph{(iv)} altogether. Simulations and comparative analysis are provided to corroborate our contributions.
\end{abstract}

\begin{graphicalabstract}
	\includegraphics{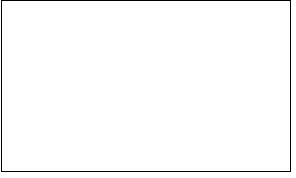}
\end{graphicalabstract}

\begin{highlights}
	\item Distributed all-time feasible strategies that solve resource allocation.
	\item Algorithm with tolerating heterogeneous bounded delays over the communication network.
	\item Convergence over uniformly-connected networks, even disconnected at some times.
	\item Addressing possible nonlinearity on the nodes or links for different applications. 
\end{highlights}


\end{frontmatter}


\section{Introduction} \label{sec_intro}
\subsection{Background}
Distributed algorithms have gained considerable attention because of recent advances in the Internet of Things (IoT), cloud computing, and parallel processing.
Distributed optimization over multi-agent networks, in particular, has emerged as an effective solution in large-scale applications ranging from machine-learning for large-scale data mining \cite{svm,khan2020optimization} to power flow optimization over the smart grid \cite{molzahn2017survey,yang2013consensus,themis2011asynchronous,cherukuri2015distributed}. 
Distributed/decentralized algorithms outperform centralized ones in many aspects: no-single-node-of-failure, scalability, speed/efficiency, etc. Distributed or decentralized optimization allows for scaling the problem-solving process to handle larger and more complex optimization problems by distributing the computational load across multiple nodes. If one of the node fails, the others can still solve the remaining objective function and this makes distributed optimization inherently more resilient to node-failures. Distributed resource allocation can use parallel processing to speed up the optimization process. By dividing the problem into sub-problems and solving them concurrently, the overall solution can be obtained faster and computationally more efficient than a centralized approach. 
In this work, we propose a distributed setup for resource allocation over general multi-agent systems with constraints on the communication network or the agents.

\subsection{Challenges}
In real-world applications,  many model nonlinearities exist in practice, some of which stem from the inherent physical constraints, e.g., actuator saturation or quantized data exchange \cite{wei2018nonlinear,Hadjicostis_run:sum}, and some are purposely added, e.g., to improve the convergence rate \cite{parsegov2013fixed,garg2019fixed2} or robustness to impulsive noise and uncertainties \cite{stankovic2019robust}. For example, in automatic generation control setup for power systems there exist ramp-rate-limit (RRL) due to limited rate of power generation by the real-world generators. In other words, the increase and decrease of the power generation cannot follow any rate but is limited. Most existing results in the literature cannot address this limit and therefore the generators cannot follow the  rates assigned by these solutions. Another example is when the information exchange among the agents is quantized. Most existing linear solutions do not address quantized information exchange among the agents and, in this aspect, the algorithm is not realistic.   
Such unseen nonlinearities make networked optimization (both constrained and unconstrained) more challenging in terms of computation, accuracy, feasibility, optimality, and convergence. 
Further, the network itself might be subject to time delays, or asynchronous data-transmission \cite{agarwal2012delay,al2020gradient,wang2018distributed,wang2019distributed,themis2011asynchronous}. The network also may lose network connectivity over some intermittent time-intervals (e.g., due to packet loss or link failure). The notion of constraint feasibility is another challenging issue. The equality-constraint ensures the resource-demand balance, and its violation may cause service disruption.
This paper aims to address general nonlinear solutions while handling network latency, uniform network connectivity, and all-time feasibility. The objective function is also nonlinear and this nonlinearity might be due to addressing the box constraints as penalty terms.


\subsection{Literature Review} 
The literature on distributed optimization, both constrained and unconstrained, mainly assumes linear models and no delay for data exchange over the network. There exist a few nonlinear reinforcement-learning-based models for resource allocation and economic dispatch \cite{hu2023distributed,dai2019distributed}. These works find optimal allocation with no prior knowledge of the mathematical formulation of the actual generation costs. In \cite{hu2023distributed} a Neural-Network is used to learn the relationship between the demand and the optimal output of each generation unit, while \cite{dai2019distributed} combines the state-action-value function approximation with a distributed optimization based on multiplier splitting.
Some existing works focus on one specific nonlinearity, e.g., quantized allocation/information-sharing \cite{fast,nedic2008distributed,rikos2020distributed,rikos_quant,wei2018nonlinear,nekouei2016convergence,nekouei2016performance} or saturated dynamics \cite{wei2018nonlinear,liu2020global}. Some others are devoted to address sign-based nonlinearities to improve the convergence rate, e.g., to reach the optimal value in finite-time  \cite{taes2020finite,chen2016distributed} or fixed-time \cite{ning2017distributed,parsegov2013fixed,garg2019fixed2}. Such fixed/finite dynamics are also prevalent in consensus literature  \cite{taes2020finite,wei2017consensus,stankovic2019robust} which also allow for robust and noise-resilient design.
However, these solutions are designed for a \textit{specific} case and cannot handle \textit{composition} of two or more nonlinearities. 
Moreover, in many applications, the network is dynamic (time-varying) due to, e.g., agents' mobility and may even get disconnected over some time intervals. Therefore, it is more practical to assume \textit{uniform connectivity}, in contrast to all-time connectivity in many existing solutions \cite{mikael2021cdc,rikos2021optimal}. Latency is another networking issue that may cause the optimization algorithm to diverge. Few works in this literature address possible \textit{homogeneous} delays at all links or asynchronous communication \cite{wang2019distributed,zhu2019distributed}, with no consideration of model nonlinearities.
In the sum-preserving constrained optimization, another challenge (other than stability) is to preserve all-time feasibility \cite{mikael2021cdc,cherukuri2015distributed} (in contrast to asymptotic feasibility \cite{aybat2016distributed,nedic2018improved,wang2019distributed}). All-time feasibility implies that, as the solution evolves, the constraint on the states always holds. For example, in the economic dispatch problem (EDP), at any termination point of the algorithm, the sum of the power states must be equal to the load demand. Otherwise, it causes disruption, power delivery issues, and even system breakdown \cite{mikael2021cdc,cherukuri2015distributed}. Such all-time feasibility conditions cannot be addressed by dual-based solutions, e.g., alternating-direction-method-of-multipliers (ADMM) methods \cite{banjac2019decentralized,falsone2020tracking,chang2016proximal,wei_me_cdc,falsone2018distributed,aybat2019distributed,dtac}. These works claim to reach feasibility fast enough within the running interval of the algorithm. Besides, the mentioned nonlinearities, network variation, uniform connectivity, and latency have not been addressed by the existing ADMM solutions \cite{banjac2019decentralized,falsone2020tracking,chang2016proximal,wei_me_cdc,falsone2018distributed,aybat2019distributed}.
Some of the existing algorithms, on the other hand,  solve specific quadratic-form problems, for example, consensus-based solutions for CPU scheduling \cite{rikos2021optimal} or economic dispatch \cite{Kar6345156}.
In general, however, the objective could be non-quadratic, e.g., because of additional penalty terms and barrier functions to handle the so-called \textit{box constraints}.

\subsection{Motivations}
What is missing in the existing literature is a general method to address ``model nonlinearitiy'' to solve (general) non-quadratic objectives while considering uniform connectivity, latency, and sum-preserving all-time feasibility altogether. The nonlinearity of the node dynamics might be due to, for example, ramp-rate-limit (RRL) of the generators in automatic generation control (AGC) \cite{hiskens_agc}.
The other common issue is quantized information exchange in real-world communication networks that motivates the use of algorithm which can handle such a nonlinearity. In this work, we kept the application general and we focus on general multi-agent systems with communication networks. The communication networks often introduce delays due to packet processing, transmission time, and network congestion. Additionally, agents are subject to some processing time before state-update or message-sharing. Therefore, the information sent from one agent/node may reach the receiver agent with certain time-delay. These delays may cause divergence of the solution and/or feasibility and optimality gap. This motivates the need to introduce a delay-tolerant solution for real-world applications. Furthermore, the network itself might be subject to some changes and lose connectivity over intermittent time-intervals due to, for example, agents' mobility. This motivates the use of algorithms to handle uniform-connectivity instead of all-time network connectivity. Another motivation is to address all-time feasibility of the equality-constraint to ensure resource-demand balance at all times. This guarantees that there is no service disruption before the termination of the algorithm at any time. This work finds many applications, e.g., in CPU scheduling \cite{ccta_quan}, coverage control \cite{MiadCons}, and plug-in electric vehicle (PEV) charging coordination \cite{falsone2017dual} among others.

\subsection{Main Contributions}
In this work, we propose two general nonlinear solutions for distributed constrained convex optimization: one addressing the nonlinearity of the nodes and one of the links. Sufficient conditions on the nonlinearity and network connectivity to ensure convergence are discussed. The Lyapunov-type proof analysis is irrespective of the specific type of nonlinearity. Therefore, the algorithm can address certain inherent physical model nonlinearity (e.g., quantization and/or saturation) or purposely added nonlinearity (e.g., signum-based) to \textit{design} fast-convergent or robust-to-noise solutions. Our model can handle the composition of more than one nonlinear mapping. In power networks, for example, the saturated generator dynamics due to RRLs is not addressed by the existing methods \cite{cherukuri2015distributed,boyd2006optimal,yang2013consensus,chen2016distributed,yi2016initialization,Kar6345156,zhu2019distributed} in the AGC setup. 
Our proposed general multi-agent network can be adapted to communicate and transmit quantized information, address actuator saturation, reach a tunable (or predefined) rate of convergence, embrace resiliency and impulsive-noise-tolerance, or any composition of such nonlinear models. We prove coupling-constraint feasibility and convergence  subject to general upper/lower sector-bound nonlinearity. In other words, sufficient conditions  on the model's nonlinearity are derived to not violate all-time sum-preserving feasibility (resource-demand balance) and convergence to the optimal value; see examples in Section~\ref{sec_sim}.

We discuss both continuous-time (CT) and discrete-time (DT) solutions. In the DT case, we design delay-tolerant solutions to handle finite \textit{arbitrary} and \textit{heterogeneous} time delays over the network. We prove constraint feasibility and
convergence/stability under bounded step sizes and certain assumptions on the time-delays. Two approaches to handle network latency are given: i) Case I, updating over a longer time scale after receiving all delayed information, and ii) Case II, updating by all the information received at the same time scale as the DT dynamics. The proposed delay-tolerant solutions lead to \textit{no feasibility gap} over switching and uniformly connected networks (instead of all-time connectivity) and in the presence of heterogeneous delays and model nonlinearity. This is in contrast to, e.g., the solution by \cite{wang2019distributed} with some feasibility gap.
The uniform connectivity is motivated by applications in mobile sensor networks, where the links may come and go as the mobile sensors (or robots) move into and out of line-of-sight (or broadcast range) of each other. The network may sometimes even lose connectivity due to link failure or even packet loss, while it maintains \textit{uniform connectivity} over some finite time intervals. Finally, we apply our proposed solutions in distributed setup with  \emph{(i)} quadratic costs and \emph{(ii)} non-quadratic costs with logarithmically quantized values. 

We summarized our contributions in the following: 
\begin{enumerate} [(i)]
	\item Possible inherent and additive nonlinearity in the model dynamics can be addressed by our proposed resource allocation model with general non-quadratic objective functions allowing for consideration of penalty/barrier functions.
	\item Our proposed delay-tolerant solution can handle possible heterogeneous and arbitrary delays of the links over general uniformly-connected networks. The values of delays might be different at the links and change over time while the network may lose connectivity over some intermittent time-intervals.
	\item Our proposed distributed resource allocation is all-time feasible implying that at all times the resource-demand balance holds even in the presence of time-delays and losing network connectivity. This is in contrast to dual formulation methods in which the feasibility gap asymptotically reduces over time. The all-time resource-demand feasibility prevents service disruption at any termination/assignment time of the algorithm.  
\end{enumerate}	
 To the best of our knowledge, no previous works in the literature address contributions \emph{(i)}-\emph{(iii)} altogether.

\subsection{Paper organization}
We introduce the preliminary notions and some useful lemmas to set up the problem in Section~\ref{sec_prob}. 
The CT solutions are proposed in Section~\ref{sec_dynamic} and their convergence is discussed in Section~\ref{sec_conv}. The DT counterparts subject to time delays are discussed in Sections~\ref{sec_DT} and \ref{sec_delay_DT}. Applications and more simulations on sparse dynamic networks are presented in Sections~\ref{sec_qra} and~\ref{sec_sim}. Finally, Section~\ref{sec_conclusion} concludes the paper. 

\subsection{General Notation} 
In this paper, we present the column vectors in bold small letters and scalars with small letters. The matrices are denoted by capital letters. For notation simplicity, $\partial f_i$  and $\partial^2 f_i$ denote $ \frac{df_i(x_i)}{dx_i}$  and $\frac{d^2f_i(x_i)}{dx_i^2}$, respectively. See the full list of notations in Table~\ref{tab_notation}.

\begin{table} [hbpt!] 
\centering
\caption{List of Notations}\
\scriptsize\setlength{\tabcolsep}{3pt}
\begin{tabular}{|c|c|}
	\hline 		\hline
	$\mb{x}$, $\mb{z}$ & column state vectors \\ \hline
	$\mb{1}_n$, $\mb{0}_n$ & column vector of ones and zeros of size $n$ \\
	\hline
	$\mb{x}^*$, $\mb{z}^*$ & optimal state vector \\ \hline
	$\widetilde{\mc{S}}_b$, ${\mc{S}}_b$ & feasibility sets (associated with $b$) of $\mb{z}$ and $\mb{x}$ states  \\
	\hline
	$x_i$, $z_i$ & element $i$ of the state vector (state of node $i$) \\\hline
	$t$, $k$ & continuous and discrete time index \\
	\hline
	$F(\cdot)$ & global cost function \\
	\hline $\mc{G}$, $\mc{V}$, $\mc{E}$ & network, set of nodes, set of links \\
	\hline
	$\nabla F(\cdot)$  & gradient of $F(\cdot)$ \\ \hline 
	$d_g$ & network diameter \\
	\hline
	$f_i(\cdot)$ & local cost function  \\ \hline $W$, $L$ & weight matrix, Laplacian matrix \\
	\hline
	$\partial f_i(\cdot)$ & first derivative of $f_i(\cdot)$  \\  \hline $\lambda_i$ & $i$th eigenvalue of $L$ \\
	\hline
	$\partial^2 f_i(\cdot)$ & second derivative of $f_i(\cdot)$ \\ \hline $\tau_{ij}$ & time-delay of link $(j,i)$ (or $j \rightarrow i$) \\		
	\hline
	$\mc{N}_i$ & set of neighbors of node $i$ \\ \hline  $\overline{\tau}$ & global upper-bound on the time-delay \\
	\hline
	$K_f$ & Lipschitz constant of $f(\cdot)$ \\ \hline $T$ & sampling step \\ 
	\hline		
	$g(\cdot)$ & nonlinear mapping \\ \hline $\mc{I}(\cdot)$ & indicator function \\
	\hline		
	$n$ & number of agents/nodes \\ \hline $|\cdot|$, $ \lceil \cdot \rceil $, $ \lfloor \cdot \rfloor $   & absolute value, ceil function, floor function  \\
	\hline			\hline
	\label{tab_notation}
\end{tabular} \normalsize
\label{tab_par}
\end{table}

The communication network of agents (also referred to as the multi-agent network) considered in this paper is modelled by a (possibly) time-varying graph $\mc{G}(t)=\{\mc{V},\mc{E}(t)\}$ with a set of time-dependent links $\mc{E}(t)$ and time-invariant set of nodes $\mc{V}=\{1,\dots,n\}$. A link $(i,j) \in \mc{E}(t)$ represents possible information exchange (communication) from the agent $i$ to the agent $j$, and further implies that the agent $i$ is in the neighbouring set of the agent $j$, defined as $\mc{N}_j(t)=\{i|(i,j)\in \mc{E}(t)\}$. The link $(j,i) \in \mc{E}(t)$ is weighted by $W_{ij}>0$ and $W(t)=[W_{ij}(t)] \in \mathbb{R}^{n \times n}_{\geq0}$ represents the weight matrix of the network $\mc{G}(t)$. Clearly, $W(t)$ follows the structure (zero-nonzero pattern) of the adjacency matrix associated with $\mc{G}(t)$. In $\mc{G}(t)$, we define a \textit{spanning tree} as a sub-graph of size $n$ (covering all nodes) in which there is only one path between every two nodes \cite{algorithm}. This spanning tree is known to include minimum possible links for connectivity \cite{algorithm}.
Define the Laplacian matrix $L(t)=[L_{ij}(t)]$ as,
\begin{align} \label{eq_laplacian}
L_{ij}(t) = \left\{ \begin{array}{ll}
	\sum_{j \in \mc{N}_i} W_{ij}(t), & \text{for}~  i=j,\\
	-W_{ij}(t), & \text{for}~i\neq j.
\end{array}\right.
\end{align}

\section{Problem Setup} \label{sec_prob}
This paper considers equality-constraint  optimization problems in the primal nonlinear formulation over a multi-agent network. The objective of the problem is to minimize the cost function while satisfying the (weighted) resource-demand balance constraint. This balance equality-constraint ensures that the weighted sum of resources meets the demand by the user, otherwise it may cause service disruption in the application under consideration. The mathematical formulation is as follows:
\begin{equation} \label{eq_dra0}
\begin{aligned}
	\displaystyle
	\min_\mb{z}
	~~ & \widetilde{F}(\mb{z}) := \sum_{i=1}^{n} \widetilde{f}_i(z_i) \\
	\text{s.t.} ~~&  \mb{z}^\top\mb{a} = b
\end{aligned}
\end{equation}
where element $z_i \in \mathbb{R}$ represents the state variable at agent $i$,
Column vector $\mb{z} = [z_1;\dots;z_n] \in \mathbb{R}^{n}$ is the collective vector state\footnote{
In its most general form, 
the problem can be extended to $\mb{z}_i \in \mathbb{R}^m$ with $m>1$ (as in \cite{doostmohammadian20211st,lakshmanan2008decentralized}). Considering $m=1$ is for the sake of simplifying the proof analysis in Sections~\ref{sec_conv} and \ref{sec_DT}.},
and column vector $\mb{a}=[{a}_1;\dots;{a}_n] \in \mathbb{R}^n$ and $b \in \mathbb{R}$ are the constraint parameters. $f_i: \mathbb{R} \mapsto  \mathbb{R}$ is the local cost (or loss) function at agent $i$, and the overall cost function is $F: \mathbb{R}^{n} \mapsto  \mathbb{R}$. 
The problem can be also extended to consider box constraints on the states, i.e., $\underline{z}_i \leq z_i \leq \overline{z}_i$. In this case, one can eliminate these extra constraints by adding proper exact (nonlinear) penalty functions to every local objective; for example, changing the local objectives to $f_i^{\sigma}(z_i) = f_i(z_i) + c([z_i - \overline{z}_i]^+ + [\underline{z}_i - z_i ]^+)$, with $[u]^+=\max \{u, 0\}^\sigma,~\sigma \in \mathbb{N}$, and $c \in \mathbb{R}^+$. 
Some other example penalties and barrier functions are discussed in \cite{bertsekas1975necessary}.
In general, such penalty (or barrier) functions are non-quadratic and nonlinear, which makes the objective function non-quadratic and nonlinear; see examples in \cite{doostmohammadian20211st,mikael2021cdc}. \textit{Feasible initialization} algorithms under such local constraints are discussed, for example, in \cite{cherukuri2015distributed}.

The problem is generally stated in the following standard sum-preserving form
\begin{equation} \label{eq_dra}
\begin{aligned}
	\displaystyle
	\min_\mb{x}
	~~ & F(\mb{x}) = \sum_{i=1}^{n} f_i(x_i) \\
	\text{s.t.} ~~&  \sum_{i=1}^{n} x_i = b
\end{aligned}
\end{equation}
which can be obtained from \eqref{eq_dra0} by simple change of variables $z_i a_i =: x_i$, with box constraints transformed into $a_i \underline{z}_i \leq  x_i \leq a_i  \overline{z}_i$ for $a_i>0$ and reversed otherwise.
We aim to provide general nonlinear dynamics to solve problem \eqref{eq_dra} over a multi-agent network for different applications. The proposed solutions need to be distributed, implying that the information available at each agent $i$ includes its own information (for example, its state and local objective) and the data received from agents $j \in \mc{N}_i$  (its direct neighbours). This work addresses the case that agents/links are constrained with some nonlinearity. Moreover, the proposed distributed solution remains feasible (i.e.,  $\sum_{i=1}^{n} x_i(t) = b$, $\forall t\geq0$) and delay-tolerant (under heterogeneous time-delays). Further, it is possible that the communication network changes over time and loses connectivity at some bounded time-intervals. The assumptions on the objective function convexity (to include possible penalty terms), the network connectivity, feasibility, and the time-delay model are essential in the problem setup, as discussed next. 

\subsection{Useful Lemmas and Definitions on  Convexity} \label{sec_lem}

\begin{definition} [Lipschitz Condition] \label{def_lips}
Let $h: \mathbb{R} \mapsto  \mathbb{R} $ represent a nonlinear mapping. Function $h(y)$ is called Lipschitz continuous if there exists a real constant $K_h$ such that for any $y_1, y_2 \in \mathbb{R}$,
\begin{align}
	|h(y_1)-h(y_2)| \leq K_h |y_1-y_2|.	 
\end{align} 
\end{definition}

\begin{definition} [Strict Convexity] A function $f:\mathbb{R} \mapsto \mathbb{R}$ is strictly convex if $\forall y_1,y_2 \in \mathbb{R},~\forall \kappa \in (0,1)$,
\begin{align}
	f(\kappa y_1+(1-\kappa)y_2)<\kappa f(y_1)+(1-\kappa)f(y_2).
\end{align}
\end{definition}

It is known that for a smooth strictly convex function, $\partial^2 f(y)> 0$
for $y \in \mathbb{R}$ \cite{bertsekas_lecture}.

\begin{ass} \label{ass_strict}
The  local cost/objective functions ${f}_i(x_i):\mathbb{R} \mapsto \mathbb{R}$, $i \in \{1,\dots,n\}$ are smooth and  strictly convex, i.e., $\partial^2 f_i(x_i) > 0$. 
\end{ass}
Note that the penalized objective function $ f_i^{\sigma}$ is also Lipschitz since $[u]^+$ is Lipschitz for $\sigma \in \mathbb{N}_{\geq 2}$.
Recalling the Taylor series expansion, the following holds.
\begin{lemma} \label{lem_strict}
\cite{boyd2006optimal} Given a continuous strictly-convex function $f(y)$, two points $y_1, y_2$, and $\Delta y =: y_1-y_2$, there exist  $\overline{y} := \kappa y_1 + (1-\kappa)y_2, 0<\kappa<1 $ such that,
\begin{align}
	f(y_1) = f(y_2) + \nabla F(y_2)^\top \Delta y +  \frac{1}{2} \Delta y^\top \partial^2 f(\overline{y})\Delta y. \label{eq_taylor}
\end{align}
\end{lemma}

\begin{lemma} \label{lem_strict2}
\cite{bertsekas_lecture}  For cost function ${f}_i:\mathbb{R} \mapsto \mathbb{R}$, assume $2 v < \partial^2 f_i < 2 u $ with $0<v \leq u <\infty$\footnote{The condition $2 v  < \partial^2 f_i(x_i) $ implies that the cost function  ${f}_i$ is strongly convex, see Assumption~\ref{ass_lips_strict} in Section~\ref{sec_rate}.}. Then,  for two points $x_1, x_2 \in \mathbb{R}$, and ${\Delta x := x_1 - x_2}$, the following statements hold:
\begin{align}
	F(x_1) &< F(x_2) + \nabla F(x_1) \Delta x + {u} \Delta x \Delta x. \label{eq_taylor1} \\ \nonumber
	F(x_1) &> F(x_2) + \nabla F(x_1) \Delta x + {v} \Delta x \Delta x. \label{eq_taylor2}
\end{align}
\end{lemma}

\begin{lemma} \label{lem_optimal_solution}
Let Assumption~\ref{ass_strict} hold. The unique optimal solution $\mb{x}^*=[{x}^*_1;\dots;{x}^*_n]$ to problem \eqref{eq_dra} is in the form,
\begin{align}
	\nabla {F}(\mb{x}^*) = \varphi^*\mb{1}_n,
\end{align}
with $\mb{1}_n$ as the column vector of $1$s, $\partial f_i \in \mathbb{R} $ as the gradient of the local  function $f_i(\cdot)$, $\nabla F(\mb{x}^*) = [\partial f_1(\mb{x}^*_1); \dots; \partial f_n(\mb{x}^*_n)] \in \mathbb{R}^{n}$, and $\varphi^* \in \mathbb{R}$.
\end{lemma}
\begin{proof}
The proof follows from the KKT condition and the Lagrange multipliers method \cite{bertsekas_lecture,boyd2006optimal}.
\end{proof}
Similarly, the unique solution of problem \eqref{eq_dra0} is in the form $\nabla \widetilde{F}(\mb{z}^*) = \widetilde{\varphi}^* \mb{a}$. Note that, in the presence of the box constraints, the above lemma assumes that $\mb{z}^*$ meets those constraints, i.e., $\underline{z}_i \leq z^*_i \leq \overline{z}_i$ for all $i$.

Throughout the paper, we refer to the feasibility condition as described below.

\begin{definition} [Feasibility]
Define the feasible sets $\widetilde{\mc{S}}_b = \{\mb{z} \in \mathbb{R}^{n}|\mb{z}^\top\mb{a} = b\}$ in \eqref{eq_dra0} or $\mc{S}_b = \{\mb{x} \in \mathbb{R}^{n}|\mb{x}^\top\mb{1}_n = b\}$ in \eqref{eq_dra}. A feasible solution then is defined as $\mb{x} \in \mc{S}_b$ or $\mb{z} \in \widetilde{\mc{S}}_b$.
\end{definition}

\begin{lemma}\label{lem_unique_feasible}
Let Assumption~\ref{ass_strict} hold. Initializing from any feasible set $\widetilde{\mc{S}}_b$ there is only one unique point $\mb{z}^* \in \widetilde{\mc{S}}_b$  such that $\nabla \widetilde{F}(\mb{z}^*) = \widetilde{\varphi}^* \mb{a}$ for  $ \widetilde{\varphi}^* \in \mathbb{R}$. Similarly, there is unique  $\mb{x}^* \in \mc{S}_b$  such that $\nabla F(\mb{x}^*) = \varphi^* \mb{1}_n$.
\end{lemma}

\begin{proof}
The proof follows from the strict convexity of the function $F(\mb{x})$ (and $\widetilde{F}(\mb{x})$) based on Assumption~\ref{ass_strict}. For detailed proof based on level-set analysis see \cite{fast,doostmohammadian20211st}.
\end{proof}


\subsection{Recall on Graph Theory} \label{sec_graph}
\begin{ass} \label{ass_G}
The following assumptions hold\footnote{In this paper, we generally assume that the network is undirected for the delay-tolerant case. Extension to \textit{balanced directed graphs}
	is discussed later in Remark~\ref{rem_wb}.}. 
\begin{enumerate} [(a)]
	\item Every link in  $\mc{G}(t)$ is bidirectional with the same weight on both sides at all times. This implies that the weight matrix $W(t)$ is symmetric and balanced for $t\geq0$.
	\item $\mc{G}(t)$ is \textit{uniformly connected} over time-window $B>0$ (or B-connected), implying that there exists  $B>0$ such that the (edge) union graph $\mc{G}_B(t)=\{\mc{V},\mc{E}_B(t)\}$ includes a spanning  tree for every  $t \geq 0$ where,
	\[ \mc{E}_B(t) = \bigcup_{t}^{t+B}\mc{E}(t).
	\]
\end{enumerate}
\end{ass}

The bidirectional condition in Assumption~\ref{ass_G} holds, for example,
when agents/nodes have similar broadcasting levels and their communication range is the same. Therefore, if $i \in \mc{N}_j(t)$ then $j \in \mc{N}_i(t)$ while the assigned weights are the same $W_{ij}(t)=W_{ji}(t)$ at all time. Later in the paper, this assumption is extended to weight-balanced directed networks with $\sum_{i=1}^n W_{ij}(t)= \sum_{j=1}^n W_{ij}(t)$ for some particular cases.
This B-connectivity condition in Assumption~\ref{ass_G} is considerably weaker than all-time connectivity in many works (e.g., the ADMM solutions \cite{banjac2019decentralized,falsone2020tracking,chang2016proximal,wei_me_cdc,falsone2018distributed,aybat2019distributed,dtac}). In particular, this allows $\mc{G}$ to lose connectivity over some time intervals. In other words, the connectivity requirement needs to be satisfied over longer time intervals in the case of links arbitrarily coming and going over the dynamic network. Such an assumption is known to be the least-connectivity assumption in consensus and distributed optimization literature \cite{nedic2014distributed,ren2005consensus}. An example contradicting this condition is the case when $\mc{G}(t)$ contains two separate sub-graphs $\mc{G}_1(t)$ and $\mc{G}_2(t)$ with no path between them for $t>t_0$, implying that no consensus can be achieved between the two. Recall that the sparse connectivity in Assumption~\ref{ass_G}(b) is strong enough for our algorithm to ensure convergence over the network (as shown later). 

\begin{remark} \label{rem_wb_removal}
The network conditions in Assumption~\ref{ass_G} are less restrictive than the existing literature. In particular, the weight-balanced and symmetric condition (a) relaxes the weight-stochasticity in \cite{falsone2020tracking,rikos2021optimal}, while uniform-connectivity (b) relaxes all-time connectivity in \cite{banjac2019decentralized,cherukuri2015distributed,mikael2021cdc,boyd2006optimal}. 
These further motivate analysis of \textit{dynamic} networks under link removal and packet loss, with no need to re-adjust (update) the link weights for ensuring stochasticity \cite{6426252}.
\end{remark}

The next lemma gives an intuition to relate the dispersion of the entries in $\mb{x}$ with the eigen-spectrum of $L$. This mainly follows from the Courant-Fischer
theorem and, for example, gives an estimate on the \textit{disagreement} value in consensus algorithms \cite{SensNets:Olfati04}. 

\begin{lemma} \label{lem_xLy}
Consider a symmetric Laplacian matrix $L$ (of a graph $\mc{G}$) and vector $\mb{x} \in \mathbb{R}^n$. Define the so-called dispersion state vector as $\overline{\mb{x}} =: \mb{x} - \frac{\mb{1}_n^\top \mb{x}}{n} \mb{1}_n$. 
The following statements hold:
\begin{align} \label{eq_laplace1}
	\mb{x}^\top L \mb{x} &= \overline{\mb{x}}^\top L \overline{\mb{x}},
	\\      \label{eq_laplace}
	\lambda_2 \|\overline{\mb{x}} \|_2^2 \leq \mb{x}^\top &L\mb{x} \leq \lambda_n \|\overline{\mb{x}} \|_2^2,
\end{align}
with $\lambda_n$ and $\lambda_2$ as the largest and smallest non-zero eigenvalue\footnote{Recall that in graph theory, for (undirected) connected graphs with link-weight equal to $1$, $\lambda_2$ is the second smallest eigenvalue of associated Laplacian matrix $L$   and is a real-valued positive number. It is also known as the Fiedler value or Algebraic connectivity  \cite{SensNets:Olfati04,graph_handbook}.}  of $L$, respectively. The results also hold for a weight-balanced (WB) directed graph (digraph) $\mc{G}$ by substituting $L_s = \frac{L+L^\top}{2}$ in \eqref{eq_laplace1}-\eqref{eq_laplace} with $L$ as the Laplacian of the digraph.
\end{lemma}
\begin{proof}
Following the definition of matrix $L$, vector $\mb{1}_n$ is in the null space of $L$. Using this, the proof of  \eqref{eq_laplace1} is straightforward since $ \mb{1}_n^\top \overline{\mb{x}} = 0$. 
The proof of \eqref{eq_laplace} follows from \eqref{eq_laplace1} and the positive definiteness of $L$. See more details in \cite{SensNets:Olfati04}. 
\end{proof}
\begin{corol} \label{cor_xLy}
The results of Lemma~\ref{lem_xLy}, can be extended for handling two variables $\overline{\mb{x}} =: \mb{x} - \frac{\mb{1}_n^\top \mb{x}}{n} \mb{1}_n$ and $\overline{\mb{y}} =: \mb{y} - \frac{\mb{1}_n^\top \mb{y}}{n} \mb{1}_n$ as follows:
\begin{align} \label{eq_laplace1_xy}
	\mb{x}^\top L \mb{y} &= \overline{\mb{x}}^\top L \overline{\mb{y}},
	\\      \label{eq_laplace_xy}
	\lambda_2 \overline{\mb{x}}^\top \overline{\mb{y}} \leq \mb{x}^\top &L\mb{y} \leq \lambda_n \overline{\mb{x}}^\top \overline{\mb{y}}.
\end{align}
\end{corol}

Note that, for a (edge) union graph $\mc{G}_B$, $\lambda_n \leq \lambda_{nB}$ and $\lambda_2 \leq \lambda_{2B}$.
Link addition (may) increase the algebraic connectivity of the network \cite{wang2008algebraic,SensNets:Olfati04}. Therefore, given $\mc{G}= \mc{G}_1 \cup \mc{G}_2$, we have $\lambda_2(\mc{G}) \geq \lambda_2(\mc{G}_1),\lambda_2(\mc{G}) \geq \lambda_2(\mc{G}_2)$. This can be extended, following the Assumption~\ref{ass_G}, to show $\lambda_2(\mc{G}_B(t)) \geq \lambda_2(\mc{G}(t))$. One can relate this to the fact that the algebraic connectivity (for $0$-$1$ adjacency matrices) satisfies $\lambda_2(\mc{G}) \geq \frac{1}{nd_g}$ (with $d_g$ as the network diameter) \cite[p. 571]{graph_handbook}. 

\subsection{Time-Delay Model for the DT agents} \label{sec_delay_model}
In general, multi-agent systems rely on communication networks to exchange information and coordinate actions. These communication networks may introduce delays due to packet processing, transmission time, and network congestion. These delays are typically quantized in integer multiples of the communication time step. Additionally, each node might have its processing time before sending a message or updating its state. These transmission delays between agents can introduce time delays.

Here, we define the time-delay model with $k$ denoting the discrete time index. Following the notation in \cite{Themis_delay,Themis_delay_TCNS,Themis2014delay_cdc}, by $\tau_{ij}(k)$ we denote the delay on the transmission link from agent $i$ to agent $j$ at DT step $k=\lfloor \frac{t}{T} \rfloor +1$ with $T$ as the sampling step and $ \lfloor \cdot \rfloor$ as the floor function. If a message sent from agent $j$ at time $t_j$ reaches agent $i$ before time $t_i>t_j$, then $\tau_{ij} =\lceil \frac{t_i}{T}  \rceil - \lfloor \frac{t_j}{T} \rfloor - 1$ 
(with $\lceil \cdot \rceil$ as the ceil function). Therefore, for $t_i - t_j < T$ there is no delay over the link $(j,i)$. However, typically delays are defined in discrete-time and the sent message at time-step $k$ and received before time-step $k + \tau_{ij}+1$ implies delay equal to $\tau_{ij}$.
\begin{ass} \label{ass_delay}
The followings hold for (integer) time-delay  $\tau_{ij}(k)$ on link $(j,i) \in \mc{E}(k)$:
\begin{enumerate}[(i)]
	\item $\tau_{ij}(k) \leq \overline{\tau}$, where $1 \leq \overline{\tau} < \infty$ is an integer representing maximum possible delays on the transmission links ($\overline{\tau} = 0$ means no delay). Upper-bound on  $\overline{\tau}$ guarantees no lost information and implies that the data from agent $i$ at time $k$ eventually is available to agent $j$ at step $k+\overline{\tau}+1$ (in finite number of time-steps).
	\item $\tau_{ij}(k)$ may change at different time-steps $k$ and is heterogeneous for different links, i.e., it may differ for agents and on different links. The time-delays are upper-bounded by some $\overline{\tau}$.
	\item The transmitted packets over every link are time-stamped and every agent $i$ knows the time $t_j$ at which agent $j$ sent the information over the link $(j,i)$, i.e., the delay $\tau_{ij}(k)$ is known.
	\item For a shared mutual link between $i$ and $j$, we consider the same delay for both sides, i.e., $\tau_{ij}(k)=\tau_{ji}(k)$.
	\item At every time $k$, at least one packet is delivered over the network (possibly with some delay), i.e., at any time step $k$ we have $\tau_{ij}(k) \neq 0$  for at least one pair $(i,j)$ (and subsequently $\tau_{ji}(k) \neq 0$ due to (iv)). 
\end{enumerate}	 		
\end{ass}

\begin{remark} \label{rem_delay}
Assumption~\ref{ass_delay} is less restrictive than many existing literature in distributed sensor networks. To clarify, part (i) only implies no packet loss and data dropout over the network. Part (ii) generalizes the existing literature with fixed delays at all links \cite{SensNets:Olfati04} by considering general heterogeneous delays as in \cite{6120272, Themis_delay}. In other words, the delays may differ at different links.
Part (iii) is a typical assumption in consensus literature \cite{SensNets:Olfati04} and data transmission networks for example for clock synchronization \cite{maroti2004flooding}. Part (iv) implies that both agents $i,j$ process their shared information \textit{simultaneously}\footnote{The assumption that  $\tau_{ij},\tau_{ji}$ are known to both agents $i,j$ is well-justified in information-theoretic perspective. This follows from Assumption~\ref{ass_delay}\textit{(iii)} and the assumption that when the packet data leaves the buffer it reaches the receiver with a fixed delay \cite{995554}. One may also consider the upper-bound on these delays as a more conservative approach.}.  It can be relaxed for asymmetric delays $\tau_{ij} \neq \tau_{ji}$, by considering $\max \{\tau_{ij},\tau_{ji}\}$ at both sides. This is to fulfil the feasibility condition (as discussed later in Section~\ref{sec_delay_DT}). Part (v) is only to ensure that at every time-step $k$ (at least) two nodes update their states.  
\end{remark}

Every agent needs to record its previous information at the last $\overline{\tau}$ time-steps to match them with the received information from agent $j \in \mc{N}_i$ at the next-coming time-steps.
Further, define $\mc{I}_{k-r,ij}$ as the \textit{indicator function} capturing the delay $\tau_{ij}(k) \leq \overline{\tau}$ on the link $(j,i)$ as follows,
\begin{align} \label{eq_indicator}
\mc{I}_{k,ij}(r) = \left\{ \begin{array}{ll}
	1, & \text{if}~  \tau_{ij}(k) = r,\\
	0, & \text{otherwise}.
\end{array}\right.
\end{align}
Define the \textit{temp graph} $\mc{G}^{\tau}(t) = \{\mc{V},\mc{E}^{\tau}(t)\}$ 
as the temporary graph representing the neighbourhood of agents at time $t$ (and time-step $k$) based on the delays $\tau_{ij}$. At time 
$kT-T<t \leq kT$ 
if agent $i$ receives a (possibly delayed) packet from $j$, then $(j,i) \in \mc{E}^{\tau}(t)$ for $kT \leq t < kT+T$, otherwise $(j,i) \notin \mc{E}^{\tau}(t)$. 

\begin{remark} \label{rem_delay_switch}
For switching networks with $B>0$ in Assumption~\ref{ass_G}(b), 
$(\overline{\tau}+1)T$ needs to be less than switching period of $\mc{G}(t)$. This implies that in case of losing a link $(i,j)$ due to a change in $\mc{G}(t)$, (at least) one delayed packet from agent $i$ reaches neighbouring agent $j$ before the link $(i,j)$ disappears. Therefore, for any $(i,j) \in \mc{E}(t)$ we have $(i,j) \in \mc{E}^{\tau}(t)$, and thus, $\bigcup_{t=k_1 T }^{k_1 T+ B}\mc{G}^{\tau}(t) = \bigcup_{t=k_1 T }^{k_1 T+ B}\mc{G}(t)$. In case of larger $\overline{\tau}$, (if possible)  uniform connectivity of $\mc{G}_{B^{\tau}}^{\tau}(t)$ over a time-window $B^{\tau} > B$ might be considered in Assumption~\ref{ass_G}.   
\end{remark}

\section{The Proposed Continuous-Time Solution  } \label{sec_dynamic}
\subsection{Proposed Continuous-Time Dynamics}
In this section, we provide two CT $1$st-order dynamics to solve problem \eqref{eq_dra}; the solution for problem \eqref{eq_dra0} similarly follows. Following the auxiliary results in the previous section, it is clear that any $\mb{x}^*$ for which $\partial f_i^*=\partial f_j^*$ for all $i,j$ must be invariant under the proposed dynamics. 
To account for nonlinearities, in general, two models can be considered: (i) the \textit{link-based} nonlinearities associated with every link/edge, and (ii) the \textit{node-based} nonlinearities associated with every agent/node.
The proposed continuous-time dynamics are based on \textit{local} information available at each agent $i$ and received from its neighbours $\mc{N}_i$. 

\textbf{Node-based Nonlinear Solution:}
\begin{align}
	\dot{\mb{x}}_i &= -\sum_{j \in \mc{N}_i} W_{ij} g\Big(\partial f_i -  \partial f_j\Big),
	\label{eq_sol}
\end{align}

\textbf{Link-based Nonlinear Solution:}
\begin{align}
	\dot{\mb{x}}_i &= - \sum_{j \in \mc{N}_i} W_{ij} \Big(g(\partial f_i) -  g(\partial f_j)\Big),
	\label{eq_sol2}
\end{align}
where  $W_{ij}$ and $\mc{N}_i$ are time-dependent in general (Assumption~\ref{ass_G}).
In the rest of this paper, we discuss the properties of \eqref{eq_sol}-\eqref{eq_sol2} as the solutions of problem \eqref{eq_dra}. 
For the sake of notation simplicity, we drop the dependence of~$W_{ij}$, $\mc{N}_i$, and $\partial f_i$ on~$t$ unless where needed.
For $g(x)=x$, \eqref{eq_sol} and \eqref{eq_sol2} represent the classic linear solution given in \cite{boyd2006optimal,cherukuri2015distributed}. However, unlike \cite{boyd2006optimal} (and many other papers in the literature) $W$ is not necessarily bi-stochastic but is only symmetric with positive entries. 

As compared to many existing linear dynamics proposed in the literature \cite{boyd2006optimal,gharesifard2013distributed,shames2011accelerated,doan2017scl,doan2017ccta,nedic2018improved,yi2016initialization,yang2013consensus,wang2018distributed,cherukuri2015distributed,lakshmanan2008decentralized} or ADMM-based solutions proposed in the  literature \cite{falsone2020tracking,chang2016proximal,wei_me_cdc,dtac,falsone2018distributed,aybat2019distributed}, this work addresses nonlinearity $g(\cdot)$ at the nodes or the links. This nonlinearity might be due to inherent property of the system, e.g., quantized information exchange among the nodes (see Section~\ref{sec_sim}) or the RRLs in automatic-generation-control setup (see Section~\ref{sec_qra}). In this aspect, the existing simplified linear solutions do not work properly and may result in an optimality gap. On the other hand, our nonlinear model allows for improving the convergence rate and finite/fixed-time convergence by adding sign-based functions as nonlinearity $g(\cdot)$. This cannot be addressed by the existing linear solutions \cite{boyd2006optimal,gharesifard2013distributed,shames2011accelerated,doan2017scl,doan2017ccta,nedic2018improved,yi2016initialization,yang2013consensus,wang2018distributed,cherukuri2015distributed,lakshmanan2008decentralized} or dual-based solutions \cite{falsone2020tracking,chang2016proximal,wei_me_cdc,dtac,falsone2018distributed,aybat2019distributed}.

We make the following assumption on the nonlinear mapping $g(\cdot)$.

\begin{ass} \label{ass_gW}	
	Function $g: \mathbb{R} \mapsto  \mathbb{R}$ is a  nonlinear odd mapping with $\frac{dg}{dx}\neq 0$ at $x = 0$ and $x g(x)> 0$ for $x \neq 0$, i.e., $g(x)$ is strongly sign-preserving. Further, there exists $K_g, \varepsilon > 0$ such that $K_g |x| \geq |g(x)| \geq \varepsilon|x|$, referred to as the sector-bound conditions\footnote{As we discuss later, the condition $\varepsilon|x| \leq  |g(x)|$ is needed for exact convergence (sign-preserving versus ``strongly'' sign-preserving) and $ |g(x)| \leq K_g |x|$ is only needed for the DT case. }. 
\end{ass}
Many existing nonlinearities satisfy the above assumption. As an example, a monotonically increasing Lipschitz odd function $g(\cdot)$ satisfies the conditions in Assumption \ref{ass_gW}. Some other examples are discussed in the rest of this section.

\begin{remark} \label{rem_combine-sol}
	Even though Eq.~\eqref{eq_sol} and \eqref{eq_sol2} represent separate nonlinearities at the nodes and the links, the results of this work hold for their combination (with both node and link nonlinearities satisfying Assumption~\ref{ass_gW}). An example is given in Section~\ref{sec_sim}.
\end{remark}

\subsection{Examples of Practical Nonlinearities in  Applications} \label{sec_app}

In this section, we provide some practical nonlinear functions $g(\cdot)$.  First, define $y := \partial f_i -  \partial f_j$ and the odd function $\mbox{sgn}^\nu(y)$ as
\begin{align} \label{eq_sign_v}
	\mbox{sgn}^\nu(y) := y|y|^{\nu-1},
\end{align}
where $|\cdot|$ denotes the absolute value and $\nu \geq 0$. For $\nu = 0$, \eqref{eq_sign_v} gives the well-known signum function denoted by $\mbox{sgn}(y)$ for simplicity. The following applications distinguish this work from many existing literature.

\textbf{Application I: Quantization}
\begin{align} 
	g_{l}(y) &:= \mbox{sgn}(y) \exp(g_{u}(\log(|y|))), \label{eq_quan_log}
\end{align}
with $g_{u}(y) := \delta \left[ \frac{y}{\delta}\right]$ and $\delta>0$ as the quantization level.
In order to allow quantized information
processing at agents and transmission links \cite{wei2018nonlinear}, one can substitute logarithmic  quantizer $g_{l}(\cdot)$
into \eqref{eq_sol} and \eqref{eq_sol2}. Note that uniform quantizer $g_{u}(\cdot)$ is not ``strongly'' sign-preserving but it is sign-preserving\footnote{Such uniformly quantized dynamics may result in (steady-state) residual and bias from the equilibrium, where the bias (residual) scales with the quantization level $\delta$ and gets arbitrarily close to zero for sufficiently small $\delta$. See \cite{quan_cdc2006,rikos2020distributed,rikos_quant,ccta_quan} for some discussions on quantized discrete-time consensus.}. 

\textbf{Application II: Saturation}
\begin{align} \label{eq_sat}
	g_\kappa(y) := 
	\begin{cases}
		\kappa\mbox{sgn}(y) & |y| > \kappa,\\
		y & |y| \leq \kappa,
	\end{cases}
\end{align}
with $\kappa>0$ as the saturation (clipping) level.
To account for the limited range of sensors/actuators (saturation) and restriction on transfer of analog/digital signals (signal clipping) \cite{liu2020global,binazadeh2016design}, the nonlinear function $g_\kappa(\cdot)$ can be substituted in  \eqref{eq_sol} and \eqref{eq_sol2}. 
This might be due to physical restrictions to follow the limited rate of increase or decrease in the actuation input $\dot{\mb{x}}_i$. An example in power grid applications is given in Section~\ref{sec_qra}.


\textbf{Application III: Finite/Fixed-time Convergence}
\begin{align} \label{eq_fixed}
	g_f(y) := \mbox{sgn}^{\nu_1}(y) + \mbox{sgn}^{\nu_2}(y),
\end{align}
where $0<\nu_1<1$,  $1<\nu_2$ (fixed-time) or $0<\nu_2<1$ (finite-time). Using this type of nonlinear optimization model some existing work on consensus and optimization \cite{garg2019fixed2,taes2020finite,parsegov2013fixed,fast} show that convergence can be achieved in finite/fixed-time. 
One may even adopt time-varying $\nu_1,\nu_2$ as in consensus algorithms \cite{8272410} to reach convergence over a prescribed time irrespective of the initialization and system parameters.  

\textbf{Application IV: Robustification}
\begin{align} \label{eq_robust_uni}
	g_p(y) := 
	\begin{cases}
		\frac{1 -\epsilon}{\epsilon d}\mbox{sgn}(y) & |y| > d\\
		0 & |y| \leq d
	\end{cases} 
\end{align}
with $0<\epsilon<1$, $d>0$. In order to design protocols robust to \textit{high-intensity outliers}, for example
in case of communication channels corrupted with impulsive noise \cite{wei2017consensus,stankovic2019robust}, one can replace the sign-preserving function $g_p(\cdot)$ in \eqref{eq_sol2}. A similar model can be considered to suppress impulsive actuation nonlinearities in \eqref{eq_sol}. However, since Eq.~\eqref{eq_robust_uni} is only sign-preserving (not ``strongly''), it may result in steady-state bias from the optimizer, i.e., convergence ensures $\mb{x}$ to reach a neighbourhood of $\mb{x}^*$.

Note that the sign-based nonlinear functions in Applications III-IV are mostly used in CT.  The applications of our proposed solutions are not limited to these nonlinear models, but any $g(\cdot)$ satisfying Assumption~\ref{ass_gW} might be adopted. For example, the composition of mentioned nonlinear mappings \eqref{eq_quan_log}-\eqref{eq_robust_uni} are also a valid choice for $g(\cdot)$ in \eqref{eq_sol} and \eqref{eq_sol2}, or many other sector-bound nonlinearities which satisfy Assumption~\ref{ass_gW}. 

\section{Convergence Analysis in Continuous-Time} \label{sec_conv}
In this section, we analyze the convergence of the CT dynamics~\eqref{eq_sol} and \eqref{eq_sol2} to the optimal value of constrained optimization \eqref{eq_dra}. We first check the feasibility and uniqueness of the solutions under the given dynamics and then prove the convergence.

\begin{lemma} [Feasibility]\label{lem_feasible_intime}
	Suppose that Assumptions~\ref{ass_G} and \ref{ass_gW} hold. Initializing by any $\mb{x}_0 \in \mc{S}_b$, the state of agents remain feasible under the CT dynamics~\eqref{eq_sol} and \eqref{eq_sol2}, i.e., $\mb{x}(t) \in \mc{S}_b$ for all $t>0$.
\end{lemma}
\begin{proof}
	The proof of feasibility for CT dynamics \eqref{eq_sol} is given in \cite{doostmohammadian20211st}. For \eqref{eq_sol2} similarly we have,
	\begin{align}
		\frac{d}{dt}(\mb{x}^\top\mb{1}_n) = -\sum_{i=1}^n\sum_{j \in \mc{N}_i} W_{ij} \Big(g(\partial f_i) -  g(\partial f_j)\Big). \label{eq_proof_feas}
	\end{align}
	From Assumptions~\ref{ass_G} and~\ref{ass_gW} we have,
	\begin{align} \nonumber
		W_{ij} \Big(g(\partial f_i) -  g(\partial f_j)\Big) = -W_{ji} \Big(g(\partial f_i) -  g(\partial f_j)\Big).
	\end{align}
	which implies that  $\frac{d}{dt}(\mb{x}^\top\mb{1}_n)=0$  in \eqref{eq_proof_feas} and $\mb{x}^\top\mb{1}_n$ remains time-invariant. Therefore, any feasible initialization  $\mb{x}_0^\top\mb{1}_n=b$ gives feasibility over time $\mb{x}(t)^\top\mb{1}_n=b$ and $\mb{x}(t) \in \mc{S}_{b}$ for all $t>0$. 
\end{proof}

In contrast to primal-dual solutions such as ADMM-based methods \cite{falsone2020tracking,chang2016proximal,wei_me_cdc,dtac,falsone2018distributed,aybat2019distributed}, Lemma~\ref{lem_feasible_intime} proves all-time feasibility of our proposed solution. This means that at any termination time of the algorithm the solution preserves feasibility while in the existing ADMM solutions \cite{banjac2019decentralized,falsone2020tracking,chang2016proximal,wei_me_cdc,dtac,falsone2018distributed,aybat2019distributed} there might be feasibility gap that converges to zero asymptotically. This implies that the solution by \cite{falsone2020tracking,chang2016proximal,wei_me_cdc,dtac,falsone2018distributed,aybat2019distributed} must be fast enough to gain feasibility before the termination of the algorithm.

Note that for any $\mb{x}$ satisfying $\nabla F(\mb{x}) = \varphi \mb{1}_n$, we have $\dot{x}_i = 0$ at every agent $i$. Therefore, such $\mb{x}$ is an equilibrium (invariant-state) of the solution dynamics~\eqref{eq_sol} and~\eqref{eq_sol2}. In the next theorem, we show that such  $\mb{x}^*$ is unique for both proposed dynamics. Note that we assume $\mb{x}^*$ satisfies the local box constraints (if there are any).

\begin{theorem} [Uniqueness in CT] \label{thm_tree}
	Suppose that Assumptions~\ref{ass_strict}, \ref{ass_G} and \ref{ass_gW} hold. Let $\mb{x}^*$ denote the equilibrium under CT dynamics \eqref{eq_sol} and \eqref{eq_sol2}. Then, $\nabla F(\mb{x}^*) = \varphi^* \mb{1}_n$ with $\varphi^* \in \mathbb{R}$.
\end{theorem}
\begin{proof}
	By contradiction assume  $\nabla F(\mb{x}^*) = ({\Lambda}^*_1;\dots;{\Lambda}^*_n)$ where $ \Lambda^* \neq  \varphi^*$ and for (at least) two agents $i,j$,
	\begin{align}
		\Lambda^*_i \neq \Lambda^*_j \Longleftrightarrow  \partial f_i^*\neq \partial f_j^*,
	\end{align}
	Define two agents $\alpha, \beta$ as,
	\begin{align}
		\alpha = \argmax_{q\in \{1,\dots,n\}}  {\Lambda}^*_{q},~
		\beta = \argmin_{q \in \{1,\dots,n\}}  {\Lambda}^*_q. \label{eq_beta0}
	\end{align}	
	From Assumption~\ref{ass_G}, the union graph $\mc{G}_B(t)$ is connected for every $t\geq 0$, and therefore,  there is a path in $\mc{G}_B(t)$ from agent (node) $\alpha$ to every other agent (node) including $\beta$. In this path, one can find (at least) two nodes $\overline{\alpha}$ and $\overline{\beta}$  with the set of neighbours $\mc{N}_{\overline{\alpha}}$ and $\mc{N}_{\overline{\beta}}$, respectively, such that
	\begin{align} 
		\Lambda^*_{\overline{\alpha}} \geq \Lambda^*_{\mc{N}_{\overline{\alpha}}},~ \Lambda^*_{\overline{\beta}} \leq \Lambda^*_{\mc{N}_{\overline{\beta}}},
		\label{eq_beta}
	\end{align}
	where the strict inequalities in the above hold for (at least) one node in $\mc{N}_{\overline{\alpha}}$ and one node in $\mc{N}_{\overline{\beta}}$. Therefore,
	from Assumption~\ref{ass_gW},  $\dot{x}^*_{\overline{\alpha}} < 0$ and   $\dot{x}^*_{\overline{\beta}} > 0$ over any time-window of length B (for both proposed dynamics). This contradicts our equilibrium assumption $\dot{\mb{x}}^* = \mb{0}$.
\end{proof}

In the virtue of Lemmas~\ref{lem_unique_feasible}, \ref{lem_feasible_intime}, and Theorem~\ref{thm_tree}, for a feasible initial state $\mb{x}_0 \in \mc{S}_b$ there exist only one equlibrium  $\mb{x}^*$ satisfying $\nabla F(\mb{x}^*) = \varphi^* \mb{1}_n$ under dynamics \eqref{eq_sol} and \eqref{eq_sol2}. 
In order to prove convergence to $\mb{x}^*$, the following lemma is needed.
\begin{lemma} \label{lem_sum}
	\cite{garg2019fixed2}
	For $g(\cdot) $ and $W$ satisfying Assumptions~\ref{ass_G}, \ref{ass_gW} and  $\psi_i,\psi_j \in \mathbb{R}$,
	
	\begin{align} \nonumber
		\sum_{i =1}^n \psi_i \sum_{j =1}^n W_{ij}g(\psi_j-\psi_i) = \sum_{i,j =1}^n \frac{W_{ij}}{2} (\psi_j-\psi_i)g(\psi_j-\psi_i),
	\end{align}
	Similarly,
	\begin{align} \nonumber
		\sum_{i =1}^n \psi_i \sum_{j =1}^n &W_{ij}(g(\psi_j)-g(\psi_i)) \\
		&= \sum_{i,j =1}^n \frac{W_{ij}}{2} (\psi_j-\psi_i)(g(\psi_j)-g(\psi_i). \nonumber
	\end{align}
\end{lemma}

\begin{theorem} [Convergence] \label{thm_converg}
	Suppose that Assumptions~\ref{ass_strict}, \ref{ass_G}, \ref{ass_gW} hold and $\mb{x}_0 \in \mc{S}_b$. The proposed CT dynamics \eqref{eq_sol} and \eqref{eq_sol2} converge to the optimal value of \eqref{eq_dra} denoted by $\mb{x}^*$ for which $\nabla F(\mb{x}^*) \in \mbox{span}\{\mb{1}_n\}$, or simply,
	\[\exists \varphi^* \in \mathbb{R}, ~~  \partial f_i^* =  \partial f_j^* = \varphi^*.\]
\end{theorem}
\begin{proof}
	The proof for dynamics \eqref{eq_sol} is given in the proof of \cite[Theorem~2]{doostmohammadian20211st}.  
	Similar proofs hold for dynamics \eqref{eq_sol2} with Lyapunov function $\overline{F}(\mb{x}) := F(\mb{x})-F(\mb{x}^*)$ and recalling from Lemma~\ref{lem_sum} that $\Big(\partial f_i -  \partial f_j\Big) \Big(g(\partial f_i)  -  g(\partial f_j)\Big)  > 0$.
\end{proof}

It is worth mentioning that the existing literature  \cite{boyd2006optimal,gharesifard2013distributed,shames2011accelerated,doan2017scl,doan2017ccta,yi2016initialization,yang2013consensus,wang2018distributed,cherukuri2015distributed,lakshmanan2008decentralized} mostly work over all-time connected networks, while, from Theorem~\ref{thm_tree} and~\ref{thm_converg}, our proposed solution works over uniformly-connected networks that may lose connectivity over some time-intervals. Note that from Theorem~\ref{thm_tree} to prove unique equilibrium of the solution we only need the union network $\mc{G}_B(t)$ to be connected over some interval $B$. This makes our solution applicable over mobile multi-agent systems where the network is dynamic and the connectivity might be lost temporarily. This is in contrast to many existing solutions \cite{boyd2006optimal,gharesifard2013distributed,shames2011accelerated,doan2017scl,doan2017ccta,yi2016initialization,yang2013consensus,wang2018distributed,cherukuri2015distributed,lakshmanan2008decentralized} where the network is static and/or all-time connected.

\section{The Proposed Solution in Discrete-Time} \label{sec_DT}
In this section, first, we provide the discrete-time version of  \eqref{eq_sol}-\eqref{eq_sol2} respectively as, 
\begin{align} 
	x_i(k+1) &= x_i(k) -T\sum_{j \in \mc{N}_i} W_{ij} g\Big(\partial f_i(k) -  \partial f_j(k)\Big), \label{eq_sol_efm}\\
	x_i(k+1) &= x_i(k) \nonumber \\
	&-T\sum_{j \in \mc{N}_i} W_{ij} \Big(g(\partial f_i(k))  -  g(\partial f_j(k)\Big)\Big), 
	\label{eq_sol_efm2}
\end{align}
with $k \geq 1$ and $T$ as the time-step. 


In the rest of the paper, we use a simplifying notation. Define the \textit{anti-symmetric} matrices $\mc{D}(k)=[\mc{D}_{ij}(k)]$ and $\mc{D}^g(k)=[\mc{D}^g_{ij}(k)]$, as the weighted difference of the gradients over all links, i.e.,
\begin{align}
	\mc{D}_{ij}(k) &= \partial f_i(k) -  \partial f_j(k),\\
	\mc{D}^g_{ij}(k) &= g(\partial f_i(k)) -  g(\partial f_j(k)).
\end{align}

The following theorem gives the proof of convergence assuming sufficiently small $T$ (as defined later in Section~\ref{sec_rate}).  

\begin{theorem} [DT Feasibility/Uniqueness/Convergence] \label{thm_tree_DT} 
	Under Assumptions~\ref{ass_strict}, \ref{ass_G}, \ref{ass_gW} and initializing from  a feasible point $\mb{x}_0 \in \mc{S}_b$, protocols \eqref{eq_sol_efm} and \eqref{eq_sol_efm2} converge to
	the feasible unique equilibrium point $\mb{x}^*$  in the form $\nabla F(\mb{x}^*) = \varphi^* \mb{1}_n$ (with $\varphi^* \in \mathbb{R}$) for sufficiently small (sampling) step $T<T_\lambda$ ($T_\lambda$ is defined later in \eqref{eq_Trange}).
\end{theorem}
\begin{proof}
	Following \eqref{eq_sol_efm} and \eqref{eq_sol_efm2}, we get
	\begin{align} 
		\sum_{i=1}^n x_i(k+1) &= \sum_{i=1}^n x_i(k) -\sum_{i=1}^n T \sum_{j \in \mc{N}_i}   W_{ij} g(\mc{D}_{ij}(k)). \label{eq_lem_efm1}\\
		\sum_{i=1}^n x_i(k+1) &= \sum_{i=1}^n x_i(k) -\sum_{i=1}^n T \sum_{j \in \mc{N}_i}   W_{ij} \mc{D}^g_{ij}(k). \label{eq_lem_efm2}
	\end{align}	
	For any link $(i,j)$ and $(j,i)$ in $\mc{G}(k)$, under the given assumptions we have $W_{ij}=W_{ji}$, $\mc{D}^g_{ij}(k)=-\mc{D}^g_{ji}(k)$, and $g(\mc{D}_{ij}(k))=-g(\mc{D}_{ji}(k))$, which implies that,
	\begin{align} \label{eq_sum0_1}
		\sum_{i=1}^n \sum_{j \in \mc{N}_i}   W_{ij} g(\mc{D}_{ij}(k)) =0,\\ \label{eq_sum0_2}
		\sum_{i=1}^n \sum_{j \in \mc{N}_i}   W_{ij} \mc{D}^g_{ij}(k) =0,
	\end{align}	
	and therefore,
	$\sum_{i=1}^n x_i(k+1) = \sum_{i=1}^n x_i(k)$
	for all $k > 0$. Therefore, initializing from $\mb{x}_0 \in \mc{S}_b$ we have $\mb{x}(k) \in \mc{S}_b$, and the feasibility for both DT protocols \eqref{eq_sol_efm_delay} and \eqref{eq_sol_efm_delay2} follows. It should be mentioned that the feasibility condition for the link-based nonlinear solution holds for general weight-balanced directed networks. The proof of uniqueness follows similar procedure as in Theorem~\ref{thm_tree}, where instead of having $\dot{\mb{x}}^*_{\overline{\alpha}} < 0$ and   $\dot{\mb{x}}^*_{\overline{\beta}} > 0$ in the CT case, we have ${\mb{x}^*_{\overline{\alpha}}(k+1)- \mb{x}^*_{\overline{\alpha}}(k)< 0}$ and   ${\mb{x}^*_{\overline{\beta}}(k+1)-\mb{x}^*_{\overline{\beta}}(k) > 0}$ over any time-window of length $B$. 
	The proof of convergence and the bound on $T$ is discussed later in Section~\ref{sec_rate}.
\end{proof}

\begin{remark} \label{rem_wb}
	For the link-based protocol~\eqref{eq_sol_efm2}, one can extend the solution to weight-balanced (WB) directed networks. This is because in the proof of Theorem~\ref{thm_tree_DT}, one can restate Eq. \eqref{eq_sum0_2} for $\sum_{i=1}^n W_{ij}= \sum_{j=1}^n W_{ji}$ over WB digraphs that proves feasibility and convergence as Lemma~\ref{lem_sum} also holds for WB digraphs. Note that proof of uniqueness is irrespective of the weights and only depends on the network connectivity. Similar reasoning proves feasibility in Lemma~\ref{lem_feasible_intime} and convergence for CT dynamics~\eqref{eq_sol2} over WB digraphs. 
	Distributed (weight) balancing algorithms can be adopted to design such directed networks \cite{6877847}.
\end{remark} 

\subsection{Rate of Convergence} \label{sec_rate}
Next, to determine the rate of convergence of the proposed (delay-free) protocols, we further make the following assumption.
\begin{ass} \label{ass_lips_strict} 
	For the local cost functions ${f}_i(x_i):\mathbb{R} \mapsto \mathbb{R}$,
	\begin{itemize}
		\item there exists $u < \infty$ such that $\partial^2 f_i(x_i) < 2u$. This implies that $\partial f_i(x_i)$ are Lipschitz continuous.
		\item (Strong-convexity) there exists $v > 0$ such that $\partial^2 f_i(x_i) > 2 v $.
	\end{itemize}
\end{ass}
To incorporate a non-smooth penalty term or barrier function $f_i^\sigma$ for the box constraints \cite{doostmohammadian20211st,mikael2021cdc}, e.g.,  $[u]^+=\max \{u,0\}^\sigma$ for $\sigma = 1$, one can replace it with smooth equivalent~${\mc{L}(u,\mu)=\frac{1}{\mu}\log (1+\exp(\mu u))}$. This is a typical reformulation in machine learning literature \cite{svm,slp_book}. It is known that~$ \mc{L} (z,\mu)-\max\{z,0\} \leq \frac{1}{\mu}$ and the two functions can become arbitrarily close by sufficiently large $\mu$~\cite{slp_book}. $[u]^+$ with $\sigma \geq 2$  is another alternative smooth option \cite{nesterov1998introductory}.
Assumption~\ref{ass_lips_strict} helps to explicitly derive a \textit{sufficient} bound on the sampling step to ensure convergence. Based on this assumption, we reformulate Lemma~\ref{lem_strict2} as follows.


\begin{lemma} \label{lem_cortes} 
	Assume the function ${f}_i:\mathbb{R} \mapsto \mathbb{R}$ to be strongly convex and $2 v <  \partial^2 f_i(x_i) < 2 u$ (Assumption~\ref{ass_lips_strict}). Then,  for two feasible points $\mb{x}(k+1), \mb{x}(k) \in \mc{S}_b$, and ${\Delta \mb{x} =: \mb{x}(k+1)-\mb{x}(k)}$,
	\begin{align} \label{eq_cortes1_1}
		F(\mb{x}(k+1)) &> F(\mb{x}(k)) +  \overline{\nabla  F}_1^\top \Delta \mb{x}
		+ v  \Delta \mb{x}^\top \Delta \mb{x}, \\
		F(\mb{x}(k+1)) &> F(\mb{x}(k)) - \frac{1}{4v}  \overline{\nabla  F}_1^\top \overline{\nabla  F}_1,
		\label{eq_cortes1}
	\end{align}
	and similarly,
	\begin{align} \label{eq_cortes2_1}
		F(\mb{x}(k+1)) &< F(\mb{x}(k)) + \overline{\nabla  F}_1^\top \Delta \mb{x}
		+ u  \Delta \mb{x} ^\top\Delta \mb{x}, \\
		F(\mb{x}(k+1)) &< F(\mb{x}(k)) - \frac{1}{4u}  \overline{\nabla  F}_1^\top \overline{\nabla  F}_1,
		\label{eq_cortes2}
	\end{align}
	with gradient dispersion defined as $\overline{\nabla  F}_1 := \nabla F (\mb{x}(k+1))  - \frac{1}{n}(\mb{1}_n^\top \nabla  F (\mb{x}(k+1)))  \mb{1}_n$.
\end{lemma}

\begin{proof}
	The proof follows from Lemma~\ref{lem_strict2}. See details in \cite{cherukuri2015distributed}.
\end{proof}
\begin{corol}
	As a direct result of Lemma~\ref{lem_cortes}, substituting $\mb{x}_2 = \mb{x}^*$ in \eqref{eq_cortes1} and \eqref{eq_cortes2}, for any feasible $\mb{x} \in \mc{S}_b$ and $\overline{F}(k) := F(k)-F^*$,
	\begin{align} \label{eq_4v}
		\frac{1}{4u} \overline{\nabla  F}^\top \overline{\nabla  F}   \leq \overline{F}(\mb{x}) \leq \frac{1}{4v} \overline{\nabla  F}^\top \overline{\nabla  F}.
	\end{align}
\end{corol}
The above corollary follows from the strong convexity of $\overline{F}$ and holds for general allocation dynamics; for example, a similar statement is given for the linear solution in \cite{cherukuri2015distributed}. 
Using Assumption \ref{ass_lips_strict} and Lemma~\ref{lem_cortes}, one can find the 
limit on the convergence rate under the proposed CT and DT protocols (assuming no latency, i.e., $\overline{\tau} =0$). Recall that, for the CT case,  $\dot{\overline{F}} =  \nabla F^\top \dot{\mb{x}}$.
In the following, we state the results for \eqref{eq_sol} and the solution for \eqref{eq_sol2} similarly follows. 
Using the definition of Laplacian in \eqref{eq_laplacian}, one can rewrite \eqref{eq_sol} (with some abuse of notation) as $\dot{\mb{x}} = - L g(\nabla F)$. Substituting this for $\dot{\mb{x}}$ and from Assumption~\ref{ass_gW} and Corollary~\ref{cor_xLy}, 
\begin{align}
	- K_g  \nabla F^\top L \nabla  F \leq \dot{\overline{F}}(\mb{x}) \leq  - \varepsilon  \nabla  F^\top L \nabla  F,
\end{align}
where we used the following inequalities from Assumptions~\ref{ass_gW},
\begin{align} \label{eq_g_epsilon1}
	\varepsilon |(\partial f_i(k))| &< |g(\partial f_i(k))| < K_g |(\partial f_i(k))|,  \\
	\varepsilon  |\mc{D}_{ij}(k)| &<|g(\mc{D}_{ij}(k))| < K_g |\mc{D}_{ij}(k)|. \label{eq_g_epsilon2}
\end{align}
For now, assume $B=0$ and later we extend it to $B>0$. Using \eqref{eq_laplace} with $\lambda_2$ described in Lemma~\ref{lem_xLy} and recalling the notation $\overline{\nabla  F}$ in Lemma~\ref{lem_cortes}, we have
\begin{align}
	-  K_g  \lambda_n  \overline{\nabla  F}^\top  \overline{\nabla  F} \leq \dot{\overline{F}} \leq  -  \varepsilon  \lambda_2  \overline{\nabla  F}^\top  \overline{\nabla  F},
\end{align}
and from \eqref{eq_4v}
\begin{align} \label{eq:Fdot}
	-4u K_g  \lambda_n \overline{F}(\mb{x}) \leq \dot{\overline{F}} \leq    - 4v \varepsilon  \lambda_2 \overline{F}(\mb{x}),
\end{align}
which is consistent with the linear case ($\varepsilon = K_g = 1$) in \cite{cherukuri2015distributed}. Recall that for $\overline{F}=0$ we also have  $\dot{\overline{F}}=0$. 
Next, for the DT case, following 
Eq.~\eqref{eq_cortes2_1},
\begin{align} \nonumber
	\overline{F}(k+1) &\leq \overline{F}(k) +
	\nabla F(k)^\top \Delta \mb{x}  + u  \Delta \mb{x}^\top  \Delta \mb{x}, 
	\nonumber
\end{align}
with $\Delta \mb{x} :=  \mb{x}(k+1) - \mb{x}(k)$. To satisfy $\overline{F}(k+1) \leq \overline{F}(k)$, we need,
\begin{align} \nonumber
	\nabla F(k)^\top \Delta \mb{x}  + u \Delta \mb{x}^\top \Delta \mb{x}  \leq 0.
\end{align} 
From Assumption~\ref{ass_gW} and following a similar line of reasoning to get Eq.~\eqref{eq:Fdot}, the above is satisfied if,
\begin{align} \nonumber
	&-\varepsilon T \lambda_2 \overline{\nabla  F}^\top \overline{\nabla  F} + u K_g^2 T^2  \overline{\nabla  F}^\top L^\top L \overline{\nabla  F}  \leq 0, 
\end{align}
where we substituted $ \Delta \mb{x} = L \nabla F = L \overline{\nabla  F}$ and used Lemma~\ref{lem_xLy} and Corollary~\ref{cor_xLy}. 
From the same Lemma~\ref{lem_xLy} we have $\overline{\nabla  F}^\top L^\top L \overline{\nabla  F} \leq \lambda_n^2 \overline{\nabla  F}^\top \overline{\nabla  F}$. Therefore, the sufficient condition for convergence is,
\begin{align} \label{eq_lambda_rate0}
	&(u K_g^2 T \lambda_n^2  - \varepsilon \lambda_2)  \overline{\nabla  F}^\top \overline{\nabla  F} \leq 0,
\end{align}
which gives the sufficient bound on $T$ as, 
\begin{align} \label{eq_Trange}
	T \leq \frac{ \varepsilon \lambda_2}{u K_g^2 \lambda_n^2} =: T_\lambda.
\end{align} 
Then, the upper bound on the convergence rate of the residual follows from Lemma~\ref{lem_cortes} and \eqref{eq_4v} as,
\begin{align} \label{eq_FF}
	\frac{\overline{F}(k+1)}{\overline{F}(k)} &\leq  1  +4 v(u K_g^2 T^2 \lambda_n^2  -T \lambda_2 \varepsilon).
\end{align}

The above gives an estimate on the (linear) convergence rate of $\overline{F}$ for sufficiently small $T$ satisfying \eqref{eq_Trange} and also proves the convergence in Theorem~\ref{thm_tree_DT}. For the quadratic cost functions (as in the economic dispatch problem), we have $u=v$, and the above equations can be more simplified. Further, in case of B-connectivity instead of all-time connectivity ($B=1$), one can derive Eq. \eqref{eq_FF} for $\frac{\overline{F}(k+B)}{\overline{F}(k)}$ over union graph $\mc{G}_B$ and modify \eqref{eq_Trange} as $T B \leq T_\lambda$ with $\lambda_2$, $\lambda_n$ as the eigenvalues of $\mc{G}_B$.

\begin{remark} \label{rem_directed}
	Recall that, from Lemma~\ref{lem_xLy} for  WB directed graphs, values $\lambda_2,\lambda_n$ in \eqref{eq_lambda_rate0} (and in the subsequent equations) denote the eigenvalues of the symmetric matrix $L_s$ instead of $L$.  
\end{remark}

\section{Networks with Time-Delays \label{sec_delay_DT}}  
Next, we extend the solutions to the time-delayed case. We consider two approaches to overcome network latency according to the following remark. 

\begin{remark} \label{rem_delay2}
	Following Assumption~\ref{ass_delay}, the knowledge of every $\tau_{ij}$ is key to satisfy the feasibility condition in the delayed case, see \eqref{eq_sum0_1} and \eqref{eq_sum0_2}. Note that $\mc{D}(k-r)$ and $\mc{D}^g(k-r)$ (with $r$ as the delay) needs to be anti-symmetric. This follows from Assumption~\ref{ass_delay} and Remark~\ref{rem_delay}. In other words, every agent $i$ knows the delay $r$ to match the received $g(\partial f_j(k-r))$ with its own previous information $g(\partial f_i(k-r))$ to find $\mc{D}_{ij}^g(k-r)$, and the same holds for agent $j$ to find $\mc{D}_{ji}^g(k-r)$.
\end{remark}

\subsection{\textbf{Case I}: information-update over a longer time-scale} \label{sec_caseI}
A straightforward solution in the delayed case is to have the agents wait for $\overline{\tau}$ steps (with $\overline{\tau}$ as the maximum possible delay) until they collect (at least) one delayed packet from the neighbouring agents (before processing for the next iteration) and then update their state\footnote{Similar solution to handle delays in consensus protocols is discussed in \cite[Remark~3]{Themis_delay}.}. This alternative approach requires knowledge of $\overline{\tau}$ (or an upper bound). Obviously, due to the agents' slower update rate, this approach's convergence rate is low. In this case, we define a new time-scale $\overline{k}=\lfloor \frac{k-1}{\overline{\tau}+1} \rfloor+1$ 
(see Fig.~\ref{fig_timescale}) 
and update the state of each agent $i$ as follows,
\begin{figure}[]
	\centering
	\includegraphics[width=3.3in]{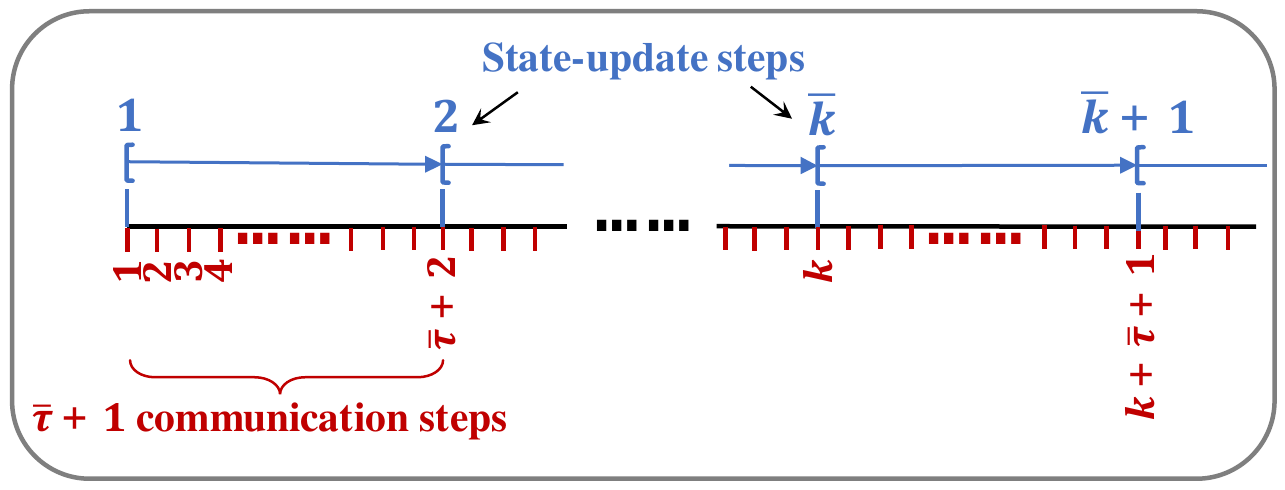}
	\caption{This figure illustrates Case I to handle delays over the network. The original communication time-scale $k$ (red) versus the state-update time-scale $\overline{k}$ (blue) are shown. The information is shared over the longer time scale $\overline{k}$ to update the state $x_i(\overline{k})$ via the dynamics~\eqref{eq_sol_efm_delay_case} and \eqref{eq_sol_efm_delay2_case}.}
	\label{fig_timescale}
\end{figure}
\begin{align} 
	x_i(\overline{k}+1) &= x_i(\overline{k})-T \sum_{j \in \mc{N}_i} W_{ij} g(\mc{D}_{ij}(\overline{k})),
	\label{eq_sol_efm_delay_case} \\
	x_i(\overline{k}+1) &= x_i(\overline{k}) -T \sum_{j \in \mc{N}_i} W_{ij} \mc{D}^g_{ij}(\overline{k}).
	\label{eq_sol_efm_delay2_case}
\end{align}
Over this $\overline{\tau}+1$ time-steps (i.e., two consecutive update steps $\overline{k}_1$ and $\overline{k}_2$) on every link $(i,j)$, every agent sends $1$ message at step $\overline{k}$ and receives the messages from $j\in\mc{N}_i$ by step $\overline{k}+1$\footnote{In an equivalent setup, agents may send their messages per scale $k$ where, at the update scale $\overline{k}$, receive at least $1$ and at most $\overline{\tau}$ messages. Then, the step size $T$ needs to be down-scaled accordingly to satisfy the convergence criteria.}, see Algorithm~\ref{alg_Dt1}.
This is used in the proof of convergence. 

\begin{algorithm}[t]
	\caption{Delay-Tolerant solution: longer time-scale} 
	\begin{algorithmic}[1]
		\State \textbf{Input:}  $W_{ij}$s, $\mc{N}_i$, $f_i(\cdot)$, $T$, $\overline{\tau}$
		\State \textbf{Initialization:} Set $\overline{k}= k = 1$, randomly allocate feasible $\mb{x}(1)$ (e.g., via \cite[Algorithm~2]{cherukuri2015distributed})
		\State Node $i$ shares $\partial f_i(1)$ with neighbors $j \in \mc{N}_i$
		\State \textbf{While} {termination criteria NOT true} \textbf{do} 
		\State Set $k \leftarrow k + 1$		
		\State \textbf{If} $ \frac{k-1}{\overline{\tau}+1}   = \overline{k} $
		\State Node $i$ calculates $\mc{D}^g_{ij}(\overline{k})$ (or $\mc{D}_{ij}(\overline{k})$) based on received data $g(\partial f_j(\overline{k}))$ (or $\partial f_j(\overline{k})$) from $j \in \mc{N}_i$
		\State Updates $x_i(\overline{k}+1)$ via Eq.~\eqref{eq_sol_efm_delay2_case} (or \eqref{eq_sol_efm_delay_case})
		\State Shares $\partial f_i(\overline{k}+1)$ with neighbors $j \in \mc{N}_i$
		\State \textbf{End If}
		\State Set $\overline{k} \leftarrow \lfloor \frac{k-1}{\overline{\tau}+1} \rfloor + 1$
		\State \textbf{End While}
		\State \textbf{Output:}  $x_i^*$, $F(\mb{x}^*) = \sum_{i=1}^{n} f_i(x_i^*)$	
	\end{algorithmic}
	\label{alg_Dt1} 
\end{algorithm}

\begin{theorem} [Feasibility/Uniqueness/Convergence] \label{thm_tree_delayI} 
	Under Assumptions~\ref{ass_strict}-\ref{ass_gW}
	and initializing from  a feasible point $\mb{x}_0 \in \mc{S}_b$, protocols \eqref{eq_sol_efm_delay_case} and \eqref{eq_sol_efm_delay2_case} converge to
	the feasible unique equilibrium point $\mb{x}^*$  in the form $\nabla F(\mb{x}^*) = \varphi^* \otimes \mb{1}_n$  for sufficiently small $T$ satisfying $ T<T_\lambda$ (with $T_\lambda$ given as \eqref{eq_Trange}).  
\end{theorem}
\begin{proof}
	The proof of feasibility and uniqueness follows a similar procedure as in the proof of Theorem~\ref{thm_tree_DT} considering the longer time-scale $\overline{k}$.
	Since at every update step $\overline{k}$ only $1$ message is received from every agent in $\mc{N}_i$, following similar expressions to derive Eq. \eqref{eq_Trange}, the same criteria $T_\lambda$ ensures convergence. 
\end{proof}


\subsection{\textbf{Case II}: using all delayed data at the same time-scale} \label{sec_caseIII}
In the most general time-delayed version of the DT protocols \eqref{eq_sol_efm} and \eqref{eq_sol_efm2}, agent $i$ updates its state based on all available (received) information as,
\begin{align} \nonumber
	x_i(k+&1) = x_i(k) \\
	&-T \sum_{j \in \mc{N}_i}  \sum_{r=0}^{\overline{\tau}} W_{ij} g(\mc{D}_{ij}(k-r)) \mc{I}_{k-r,ij}(r),
	\label{eq_sol_efm_delay} \\
	x_i(k+&1) = x_i(k). \nonumber \\
	&-T \sum_{j \in \mc{N}_i}  \sum_{r=0}^{\overline{\tau}} W_{ij} \mc{D}^g_{ij}(k-r) \mc{I}_{k-r,ij}(r),
	\label{eq_sol_efm_delay2}
\end{align}
Note that in \eqref{eq_sol_efm_delay}, $\mc{D}_{ij}(k-r)$  
can be easily calculated as agent $i$ records all its gradients $\partial f_i(k)$ at the last $\overline{\tau}$ steps 
and knows the time delay $r$ of the received \textit{time-stamped} message $\partial f_j(k-r)$. In fact, the second $\sum$ sums all the gradient differences based on all received data from $\mc{N}_i$ at step $k$, i.e.,  for all pairs of $\mc{D}_{ij}(s)$ and $s$ satisfying\footnote{Recall that the processed information at the time-step $k$ are received in the time-period $(kT-T,T]$. Therefore, \eqref{eq_s} can be rewritten as $\{k -\overline{\tau} \leq s \leq k, kT-T< sT+\tau_{ij}(s)T \leq kT, j \in \mc{N}_i\}$
	in the asynchronous case, for example, the delays could be positive real values (instead of integers).},
\begin{align} \label{eq_s}
	\{k -\overline{\tau} \leq s \leq k, s+\tau_{ij}(s)=k, j \in \mc{N}_i\}.
\end{align}
or equivalently,
\begin{align} \label{eq_s}
	\{\mc{I}_{k-s,ij}=s, j \in \mc{N}_i\}.
\end{align}
Similar arguments hold for  $\mc{D}^g_{ij}(k-r)$. 
Further, for Lipschitz $g(\cdot)$  with constant $K_g<\infty$ (and smooth $f_i(\cdot)$), from Definition~\ref{def_lips}, we have $\mc{D}^g_{ij}(k-r) \leq K_g \mc{D}_{ij}(k-r)$. 


\begin{remark}
	Following Remark~\ref{rem_delay_switch}, for switching network $\mc{G}(k)$, the switching period needs to be longer than $(\overline{\tau}+1)T$ (or $\overline{\tau}+1$ time-steps). This assumption implies that agent $i$ receives at least one (possibly) delayed packet from every neighbour in $\mc{N}_i$ before losing its link due to network variation (ensuring Assumption~\ref{ass_delay}(v)). Therefore, following the B-connectivity in Assumption~\ref{ass_G}, we need $(\overline{\tau}+1)T<B$. 
\end{remark}

Note that for the update at every step $k$, Assumption~\ref{ass_delay}\textit{(v)} says that at least one package over the network is delivered. Otherwise, $x_i(k)=x_i(k-1)$ for all $i$ and no update occurs. The solution is summarized in Algorithm~\ref{alg_Dt2}.
The feasibility/uniqueness under \eqref{eq_sol_efm_delay} and \eqref{eq_sol_efm_delay2} follows similar to the delay-free case.

\begin{algorithm}[t]
	\caption{Delay-Tolerant solution: the same time-scale} 
	\begin{algorithmic}[1]
		\State \textbf{Input:}  $W_{ij}$s, $\mc{N}_i$, $f_i(\cdot)$, $T$, $\overline{\tau}$
		\State \textbf{Initialization:} Set $k = 1$, randomly allocate feasible $\mb{x}(1)$ (e.g., via \cite[Algorithm~2]{cherukuri2015distributed})
		\State \textbf{While} {termination criteria NOT true} \textbf{do} 
		\State Node $i$ receives time-stamped $g(\partial f_j(k-\overline{\tau}_{ij}))$ (resp. $\partial f_j(k-\overline{\tau}_{ij}))$) from $j \in \mc{N}_i$
		\State Records $\mc{D}^g_{ij}(k-\overline{\tau}_{ij})$ (resp. $\mc{D}_{ij}(k-\overline{\tau}_{ij})$) based on known delay $\overline{\tau}_{ij}$
		\State Updates $x_i(k+1)$ via Eq.~\eqref{eq_sol_efm_delay2} (resp. \eqref{eq_sol_efm_delay})
		\State Shares $\partial f_j(k+1)$ with neighbors $j \in \mc{N}_i$
		\State Sets $k \leftarrow k+1$ 
		\State \textbf{End While}
		\State \textbf{Output:}  $x_i^*$, $F(\mb{x}^*) = \sum_{i=1}^{n} f_i(x_i^*)$
	\end{algorithmic}
	\label{alg_Dt2} 
\end{algorithm}

\begin{lemma} [Feasibility/Uniqueness] \label{lem_feasible_intime_delay}
	Let Assumptions~\ref{ass_strict}-\ref{ass_gW} 
	hold. Initializing from feasible states $\mb{x}_0 \in \mc{S}_b$, under DT protocols  \eqref{eq_sol_efm_delay} and \eqref{eq_sol_efm_delay2}, $\mb{x}(k) \in \mc{S}_b$. Further, the unique equilibrium point $\mb{x}^*$ satisfies $\nabla F(\mb{x}^*) = \varphi^* \mb{1}_n$. 
\end{lemma}
\begin{proof}
	Following \eqref{eq_sol_efm_delay} and \eqref{eq_sol_efm_delay2}, we get
	\begin{align} \nonumber
		\sum_{i=1}^n &x_i(k+1) = \sum_{i=1}^n x_i(k) \\
		&-\sum_{i=1}^n T \sum_{j \in \mc{N}_i}  \sum_{r=0}^{\overline{\tau}} W_{ij} g(\mc{D}_{ij}(k-r)) \mc{I}_{k-r,ij}(r). \\ \nonumber
		\sum_{i=1}^n &x_i(k+1) = \sum_{i=1}^n x_i(k) \\
		&-\sum_{i=1}^n T \sum_{j \in \mc{N}_i}  \sum_{r=0}^{\overline{\tau}} W_{ij} \mc{D}^g_{ij}(k-r) \mc{I}_{k-r,ij}(r). 
	\end{align}	
	For any link $(i,j)$ and $(j,i)$ in $\mc{G}(k)$, from Assumptions~\ref{ass_G}, \ref{ass_delay}, and \ref{ass_gW} we have $W_{ij}=W_{ji}$, $\mc{D}^g_{ij}(k-r)=-\mc{D}^g_{ji}(k-r)$,  $g(\mc{D}_{ij}(k-r))=-g(\mc{D}_{ji}(k-r))$, and  $\mc{I}_{k-r,ij}(r)=\mc{I}_{k-r,ji}(r)$ for $0 \leq r \leq \overline{\tau}$. This implies that,
	\begin{align} \nonumber
		\sum_{i=1}^n T \sum_{j \in \mc{N}_i}  \sum_{r=0}^{\overline{\tau}} W_{ij} g(\mc{D}_{ij}(k-r)) \mc{I}_{k-r,ij}(r)=0,\\ \nonumber
		\sum_{i=1}^n T \sum_{j \in \mc{N}_i}  \sum_{r=0}^{\overline{\tau}} W_{ij} \mc{D}^g_{ij}(k-r) \mc{I}_{k-r,ij}(r)=0,
	\end{align}	
	and therefore,
	$\sum_{i=1}^n x_i(k+1) = \sum_{i=1}^n x_i(k)$
	for all $k \geq 1$. Therefore, initializing from $\mb{x}_0 \in \mc{S}_b$ we have $\mb{x}(k) \in \mc{S}_b$ under both \eqref{eq_sol_efm_delay} and \eqref{eq_sol_efm_delay2}. This proves the feasibility. The proof of uniqueness follows similar reasoning as in the proof of Theorem~\ref{thm_tree_DT} and~\ref{thm_tree_delayI} over the uniformly connected network $\mc{G}_{B}(t)$ (or $\mc{G}_{B^{\tau}}^{\tau}(t)$ in Remark~\ref{rem_delay_switch}). 
\end{proof}

\begin{remark} \label{rem_delay_1slot}
	Case I is more practical than Case II in applications with low capacity/buffer at the computing nodes.
	This is because at time $k=\overline{k}$ 
	node $i$ sends (and receives) one message to (and from) the nodes $\mc{N}_i$. 
	Note that, for large delays (and the same $T$), this solution converges (although with a low rate) while Case II may not necessarily converge. 
	Further, for unknown (or large) $\overline{\tau}$, Case II requires a high-capacity memory/buffer at the nodes to record the previous information (on the gradients). 
	In case of limited ($m$-slot) memory/buffer, the nodes may record and use a portion (last $m$) of the previous states instead, and discard the rest before $k-m$ (since they are time-stamped). This simply implies losing some links over $\mc{G}_{\overline{B}}^{\tau}$
	\footnote{Convergence over lossy networks, following from Assumption~\ref{ass_G} and remark~\ref{rem_wb_removal}, is another promising direction of our future research.} and follows from Remark~\ref{rem_wb_removal} and~\ref{rem_delay_switch} and Assumption~\ref{ass_G}.
\end{remark}

\begin{theorem} [Sufficiency]
	\label{thm_converg_delay}
	Under Assumptions~\ref{ass_strict}-\ref{ass_gW}, and initializing from feasible state $\mb{x}_0 \in \mc{S}_b$,  the proposed protocols \eqref{eq_sol_efm_delay}-\eqref{eq_sol_efm_delay2} converge to the optimal solution of \eqref{eq_dra} for 
	$T(\overline{\tau}+1)<T_\lambda$.
\end{theorem}

\begin{proof}
	First, consider homogeneous delays $\tau_{ij}=\overline{\tau}$, where agents' states at any time-step $k$ get updated next at $k+\overline{\tau}+1$ and every $\overline{\tau}+1$ steps afterwards (see Fig.~\ref{fig_timescale}). For this case, following Theorem~\eqref{thm_tree_DT}, $T<T_\lambda$ ensures the convergence.  
	For general heterogeneous time-varying delays, $\Delta \mb{x}$ needs to be scaled by $\overline{\tau}+1$, and accordingly, $T_\lambda$ is down-scaled by $\overline{\tau}+1$ to guarantee convergence. 
	Consider two cases: (i) for time-invariant delays $\sum_{r=0}^{\overline{\tau}} \mc{I}_{k-r,ij}(r)=1$ in \eqref{eq_sol_efm_delay}, implying that \textit{only one} delayed packet is received from every $j \in \mc{N}_i$. This implies the same bound as $T<T_\lambda$.
	For time-varying delays, from~\eqref{eq_indicator},
	$0 <  \sum_{r=0}^{\overline{\tau}} \mc{I}_{k-r,ij}(r) \leq (\overline{\tau}+1)$.
	This implies that $\Delta \mb{x}$ is scaled by $\overline{\tau}+1$ and accordingly, from \eqref{eq_lem_efm1}, $T$ needs to be down-scaled by $\overline{\tau}+1$, i.e., $T(\overline{\tau}+1) < T_{\lambda}$.
	The proof for dynamics \eqref{eq_sol_efm_delay2} similarly follows.
\end{proof}
For time-varying but equi-probable delays $r = 0,1,\dots,\overline{\tau}$ in~\eqref{eq_indicator}, one can claim that (in average) over $\overline{\tau}+1$ trials node $i$ receives  one message from $j \in \mc{N}_i$ (at every $k$) and the bound is $T < T_{\lambda}$.

Communication time-delay over the network is not addressed by the existing solutions \cite{boyd2006optimal,gharesifard2013distributed,shames2011accelerated,doan2017scl,doan2017ccta,nedic2018improved,yi2016initialization,yang2013consensus,wang2018distributed,cherukuri2015distributed,lakshmanan2008decentralized}. However, latency is a common issue in many multi-agent systems including the distributed resource allocation setup. Therefore, most existing works in the literature may not work properly in the presence of time-delays in the data-transmission network. The mentioned literature may lose resource-demand feasibility and/or result in some optimality gap in the presence of time-delays. In this aspect, our proposed delay-tolerant model advances  the state-of-the-art and provides solutions that can be implemented in practice.



\section{Simulation in Distributed Scheduling Setup} \label{sec_qra}

In this section, we consider resource allocation over the power generation networks and the smart grid, known as the \textit{economic dispatch problem} (EDP) \cite{amjady2009economic,chen2016distributed,cherukuri2015distributed}. The objective is to allocate optimal power outputs to the electricity generators to supply the \textit{load demand} $D$ (in MW) and to minimize the operating costs. In   \cite{wood1996power,yang2013consensus,chen2016distributed,yi2016initialization,Kar6345156} this cost function is given as,
\begin{align} \label{eq_f_quad}
	F(\mb{x}) =  \sum_{i=1}^n \gamma_i {x}_i^2+ \beta_i {x}_i + \alpha_i,~\sum_{i=1}^n {x}_i = D, m_i \leq {x}_i \leq M_i,
\end{align}
where ${x}_i \in \mathbb{R}$ is the output power at generator $i$. To include the box constraints, penalty terms are considered \cite{doostmohammadian20211st}. 
The parameters in the cost function \eqref{eq_f_quad} are defined based on the type of the power generators (coal-fired, oil-fired, nuclear, etc.), for example, see Table~\ref{tab_par}. The other parameter values are: $\alpha_i = 0$ and $m_i = 20$. The min and max RRL values given in \cite{git_matpower} are equal to $1$ and $3$ $MW/min$ (for oil/coal-type generators).  

\begin{table} [hbpt!]
	\centering
	\caption{Parameters of the generator cost (in $\$$) 
		and its maximum power (in $MW$) for different types  \cite{wood1996power,yang2013consensus,Kar6345156}}
	
	\small
	\begin{tabular}{|c|c|c|c|}
		\hline
		\hline
		Type& $M_i$ & $\beta_i$ & $\gamma_i$  \\
		\hline
		A & 80 & 2.0 & 0.04 \\
		\hline
		B& 90 & 3.0& 0.03\\
		\hline
		C & 70 & 4.0& 0.035 \\
		\hline
		D& 70 & 4.0& 0.03 \\
		\hline
		E & 80 & 2.5& 0.04 \\
		\hline 		\hline		
	\end{tabular} \normalsize
	\label{tab_par}
\end{table}

In a more complicated setup, such a quadratic cost model is defined to adjust the power-demand mismatch over the grid (e.g., due to generator outage), known as the \textit{Automatic Generation Control (AGC)} problem \cite{doostmohammadian20211st}. Given a known power mismatch, the idea is to allocate enough power to the generators to compensate for it while
minimizing the power deviation cost. In this setup, the generators are subject to an extra physical constraint, known as ramp-rate-limit (RRL). This implies that the rate of increase or decrease of the produced power is constrained within certain limits and the generators cannot freely speed up/down their power generation. 
Such nonlinear constraints are a determining factor on the stability of the grid \cite{hiskens_agc}. The solutions of the existing linear methods \cite{boyd2006optimal,gharesifard2013distributed,cherukuri2015distributed} may assign any rate of change $\dot{\mb{x}}_i$ to the power generators. Thus, they cannot address RRLs,  which, in reality, the solutions cannot be followed by the generators. This may fail the feasibility and lead to sub-optimality. In contrast, considering the nonlinearity $g(\cdot)$ as saturation function \eqref{eq_sat}, the proposed nonlinear solutions in this work can address RRL constraints where the limits can be tuned by the parameter $\kappa$. A comparing example is given next.

\subsection{Allocation with No Delay}
Assume $n=50$ generators with the supply-demand constraint $D=3200~MW$. Initially, it is considered that $x_i=\frac{D}{n}=64~MW$ at every generator $i$. The parameters in~\eqref{eq_f_quad} are randomly set from Table~\ref{tab_par} using MATLAB \texttt{randi} function.
For power allocation, as an academic example, we consider protocol~\eqref{eq_sol_efm} with saturated nonlinearity
$g_\kappa(\cdot)$ with level $\kappa=\frac{1}{60}$ on the node dynamics and protocol \eqref{eq_sol_efm2} with $g_f(\cdot)$ as the sign-based function \eqref{eq_fixed} (with $\nu_1=0.4$, $\nu_2=1.6$) on the links. The network is considered as a random Erdos-Renyi graph (with link probability $p=0.2$) with random symmetric link weights $0.005 \leq W_{ij} \leq 0.025$. We provide a comparative simulation analysis to support our results. Both solutions are compared with some state-of-the-art solutions in the literature in Fig.~\ref{fig_compare} for $T=1$; namely: linear \cite{boyd2006optimal}, accelerated linear  \cite{shames2011accelerated}, finite-time  \cite{chen2016distributed}, and single-bit \cite{taes2020finite} solutions. For the single-bit protocol \cite{taes2020finite} we decreased the link weights by $80\%$ to reduce the chattering effect in Fig.~\ref{fig_rrl}. Note that the RRL constraint is only met by our saturated solution as shown in Fig.~\ref{fig_rrl}. In other words, in a real scenario the generators cannot follow the iterative solution by \cite{boyd2006optimal, shames2011accelerated, chen2016distributed,taes2020finite} since their power generation rates at some intervals violate the RRLs. This may either cause feasibility gap or optimality gap in real world. But our saturated-based solution admits the RRLs. The min RRL value equal to $1$ $MW/min$ (or $\frac{1}{60}$ $MW/sec$) is considered that meets the requirement by all the generators\footnote{Even though we used the same min RRL in this simulation, this parameter can be further tuned for different generators in dynamics \eqref{eq_sol}-\eqref{eq_sol2} via the link weights and the node degrees.}. 

\begin{figure}[]
	\centering
	\includegraphics[width=3.2in]{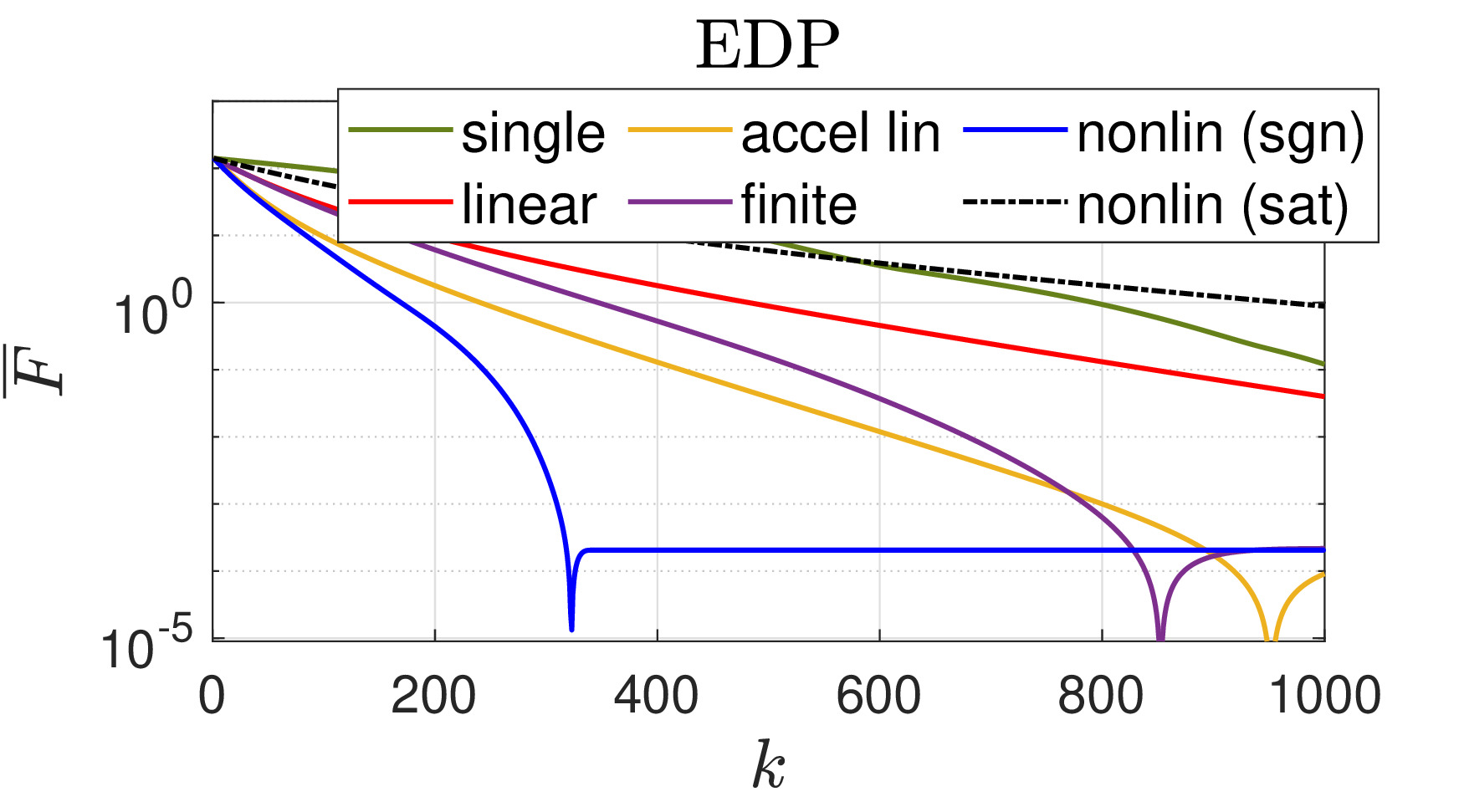}
	\caption{The performance of our node-based protocol~\eqref{eq_sol_efm} subject to saturation and link-based protocol~\eqref{eq_sol_efm2} with sign-based nonlinearity compared with linear \cite{boyd2006optimal}, accelerated linear  \cite{shames2011accelerated}, finite-time  \cite{chen2016distributed}, and single-bit \cite{taes2020finite} protocols. }
	\label{fig_compare}
\end{figure}

\begin{figure}[]
	\centering
	\includegraphics[width=3.2in]{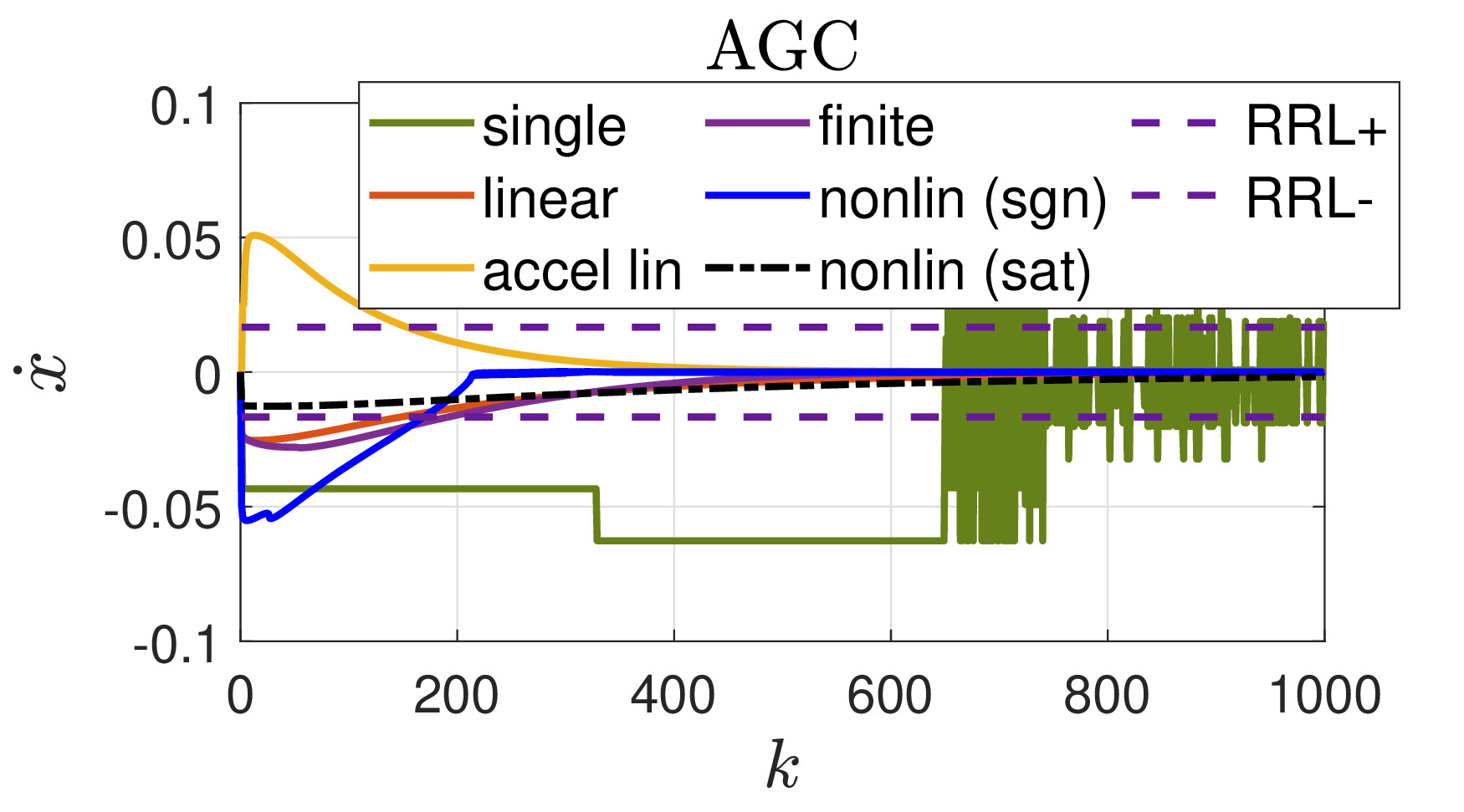}
	\caption{The rate of generated power $\dot{x}_i$ (in $MW/sec$) at a sample generator $i=$ are compared under different solutions. Only our proposed saturated nonlinear model admits the RRL constraints on the generator dynamics (shown by dashed lines). }
	\label{fig_rrl}
\end{figure}

Box constraints ${m_i \leq x_i \leq M_i}$  on the generators are considered by adding smooth penalty terms with $\sigma=2$. We modify the objective as $ f_i^{\sigma} = f_i(x_i) + c([x_i - M_i]^+)^2 + c([m_i- x_i]^+)^2$ with $c = 1$. The evolution of power states under the sign-based dynamics is shown in Fig.~\ref{fig_x_box}, where both feasibility and power limit constraints are met. Note that our given assumptions ensure that at every time step the deviated power at nodes on two sides of every link is balanced such that the feasibility constraint is satisfied (i.e., the generated power equals the demand) at all times. 
\begin{figure}[]
	\centering
	\includegraphics[width=3.2in]{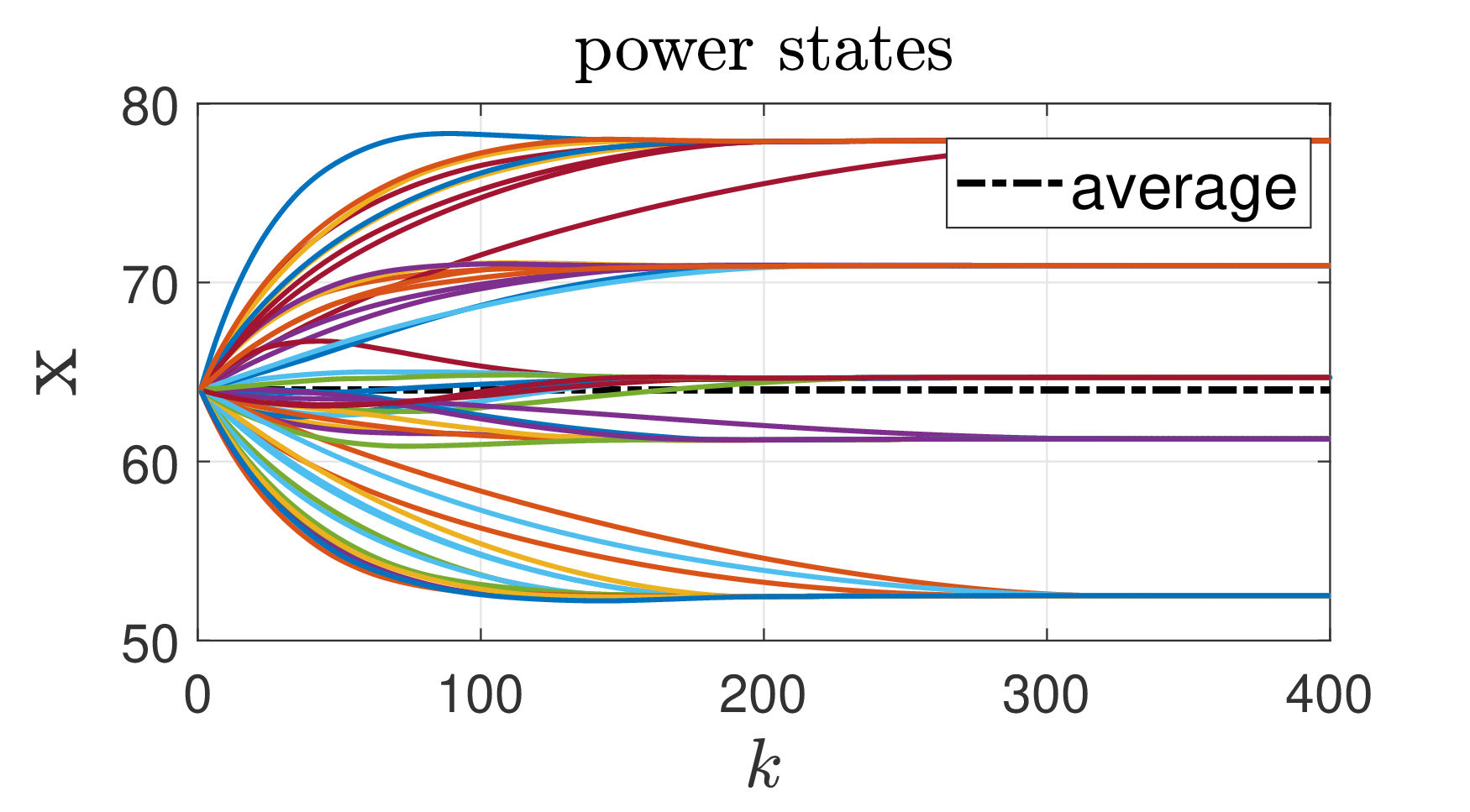}
	\includegraphics[width=3.2in]{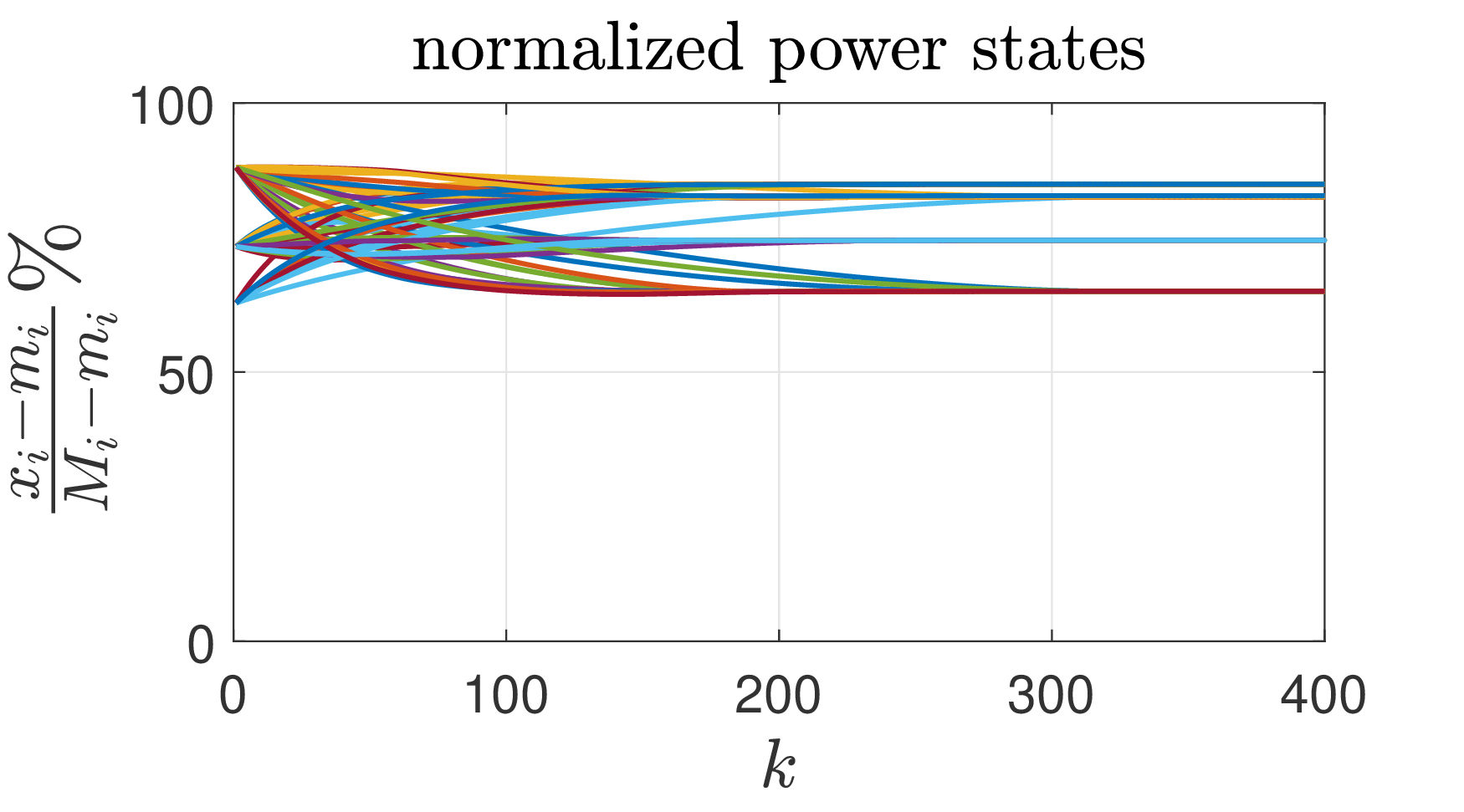}
	\caption{The allocated power states $x_i$ (in $MW$) under sign-based dynamics and box constraints on the output powers as ${m_i \leq x_i \leq M_i}$. The feasibility is checked by the average of output powers which is constant (black dashed line). The normalized allocated power states $x_i$ are also shown to verify the box constraints.}
	\label{fig_x_box}
\end{figure}

To compare the computational complexity of the algorithms, the  elapsed times (for one iteration) are given in Table~\ref{tab_time}. The simulation is done with MATLAB R2021b Intel Core i5 @ 2.4GHz processor RAM 8GB. 
\begin{table} [hbpt!]
	\centering
	\caption{The elapsed time (in ms) to run one iteration (in average) for each scenario in Fig.~\ref{fig_compare}. This table gives a measure to compare the computational complexity of the algorithms.}
	
	\small
	\begin{tabular}{|c|c|c|c|c|c|}
		\hline	\hline
		lin & accel  & single & finite   & sat & sgn \\
		\hline
		0.37 & 0.31 & 0.38 & 0.75 & 0.38 & 1.18 \\
		\hline 	\hline
	\end{tabular} \normalsize
	\label{tab_time}
\end{table}

In Table~\ref{tab_time2}, we compared the number of iterations to reach (a predefined) cost residual $\overline{F} = 1$. For example, within this range around the optimal value the solution is considered good enough. Note that our proposed sign-based solution, as compared to other solutions, although having more computational complexity converges very fast and within fewer number of iterations.

\begin{table} [hbpt!]
	\centering
	\caption{The number of iterations $k$ to reach  $\overline{F} = 1$.}
	\small
	\begin{tabular}{|c|c|c|c|c|c|}
		\hline	\hline
		lin & accel  & single & finite & sat & sgn \\
		\hline
		480 & 239 & 791 & 345 & 962 & 168  \\
		\hline 	\hline
	\end{tabular} \normalsize
	\label{tab_time2}
\end{table}

Next, we study the network density on the number of iterations to reach a certain residual (as the termination criteria of the algorithm). In Fig.~\ref{fig_resid}, this termination point is $\overline{F} = 0.01$ and the ER link probability is changed from $15\%$ to $45\%$. The link weights $W_{ij}$ are chosen randomly in the range $[0.02~0.12]$. The simulation is averaged over $5$ Monte-Carlo (MC) trials with $T=0.05$.

\begin{figure}[]
	\centering
	\includegraphics[width=1.6in]{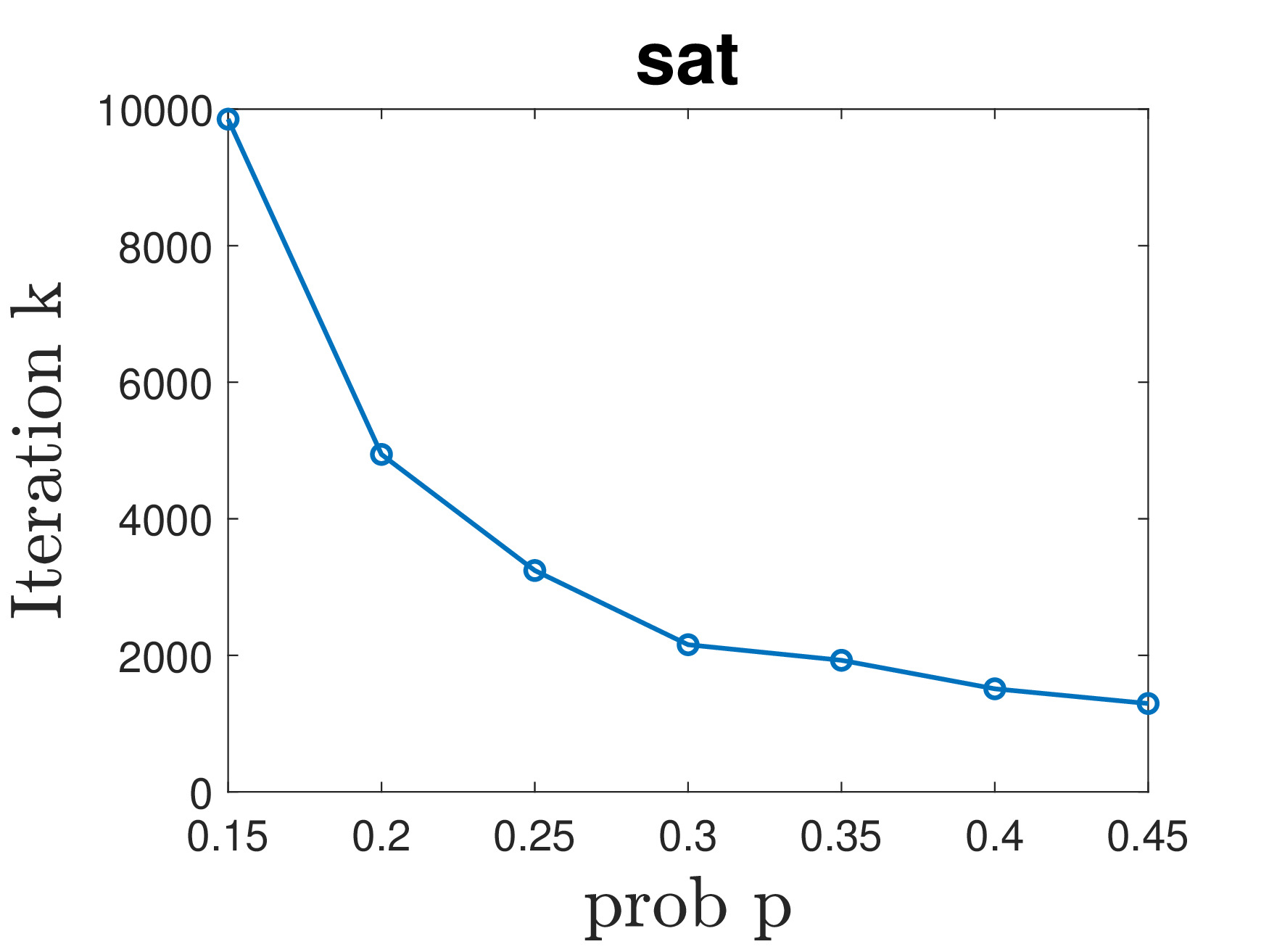}
	\includegraphics[width=1.6in]{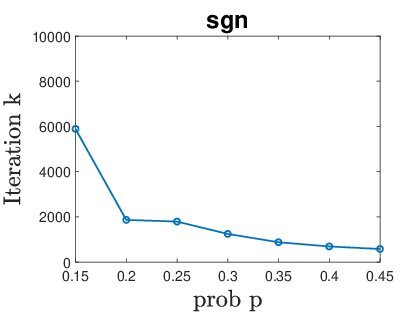}
	\caption{The number of iterations to reach residual $\overline{F} = 10^{-2}$ versus ER network connectivity tuned by the link probability $p$ for sat (level $\kappa=1$) and sgn  ($\nu_1=0.4$, $\nu_2=1.6$) nonlinearity. }
	\label{fig_resid}
\end{figure}

The convergence time/iterations to reach $\overline{F} = 10^{-2}$ versus network size $n$ is shown in Fig.~\ref{fig_resid_size}. The parameters are as follows: $p=30\%$ as the ER link probability, $W_{ij} \in [0.05~0.2]$, and $5$ MC trials.  Two different scenarios are considered for $T$: diminishing by size as $T=0.05(1-\frac{n-20}{100})$ and fixed step size $T=0.025$.

\begin{figure}[]
	\centering
	\includegraphics[width=1.53in]{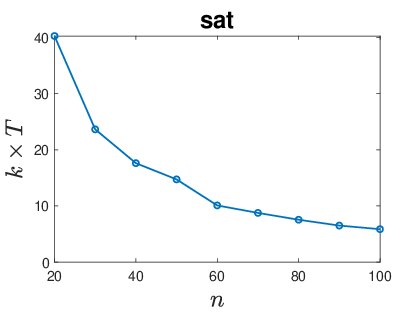}
	\includegraphics[width=1.53in]{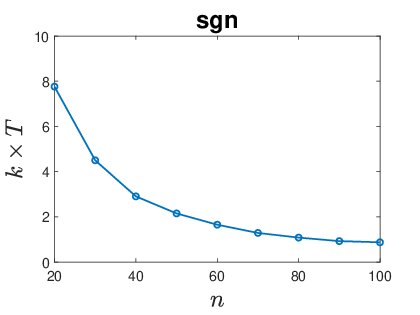}
	\includegraphics[width=1.6in]{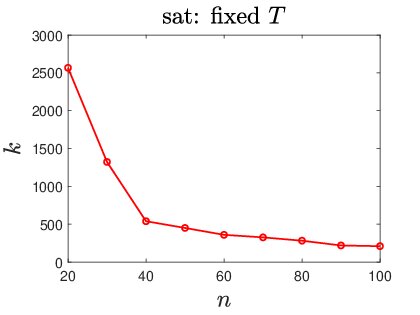}
	\includegraphics[width=1.6in]{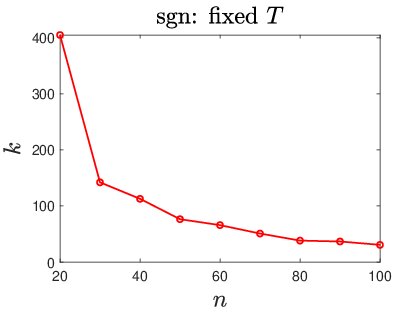}
	\caption{The convergence time/iterations to reach residual $\overline{F} = 0.01$  versus ER network size $n$ for sat (level $\kappa=2$) and sgn  ($\nu_1=0.5$, $\nu_2=1.5$) nonlinearity. (Top) diminishing step-sizes as size increases, (Bottom) fixed step-sizes. }
	\label{fig_resid_size}
\end{figure}

In Fig.~\ref{fig_resid_p_T}, for the same termination criteria $\overline{F} = 10^{-2}$, the number of iterations for different step sizes $T$ and link densities $p$ are given. As we see in the figure, the convergence iteration multiplied by the step size is almost constant irrespective of the network connectivity (and eigen-spectrum) and the step size $T$. Note that the $T$ values are chosen small enough for convergence, and for large values of $T$ violating Eq.~\eqref{eq_Trange} the solution may not necessarily converge. 

\begin{figure}[]
	\centering
	\includegraphics[width=1.6in]{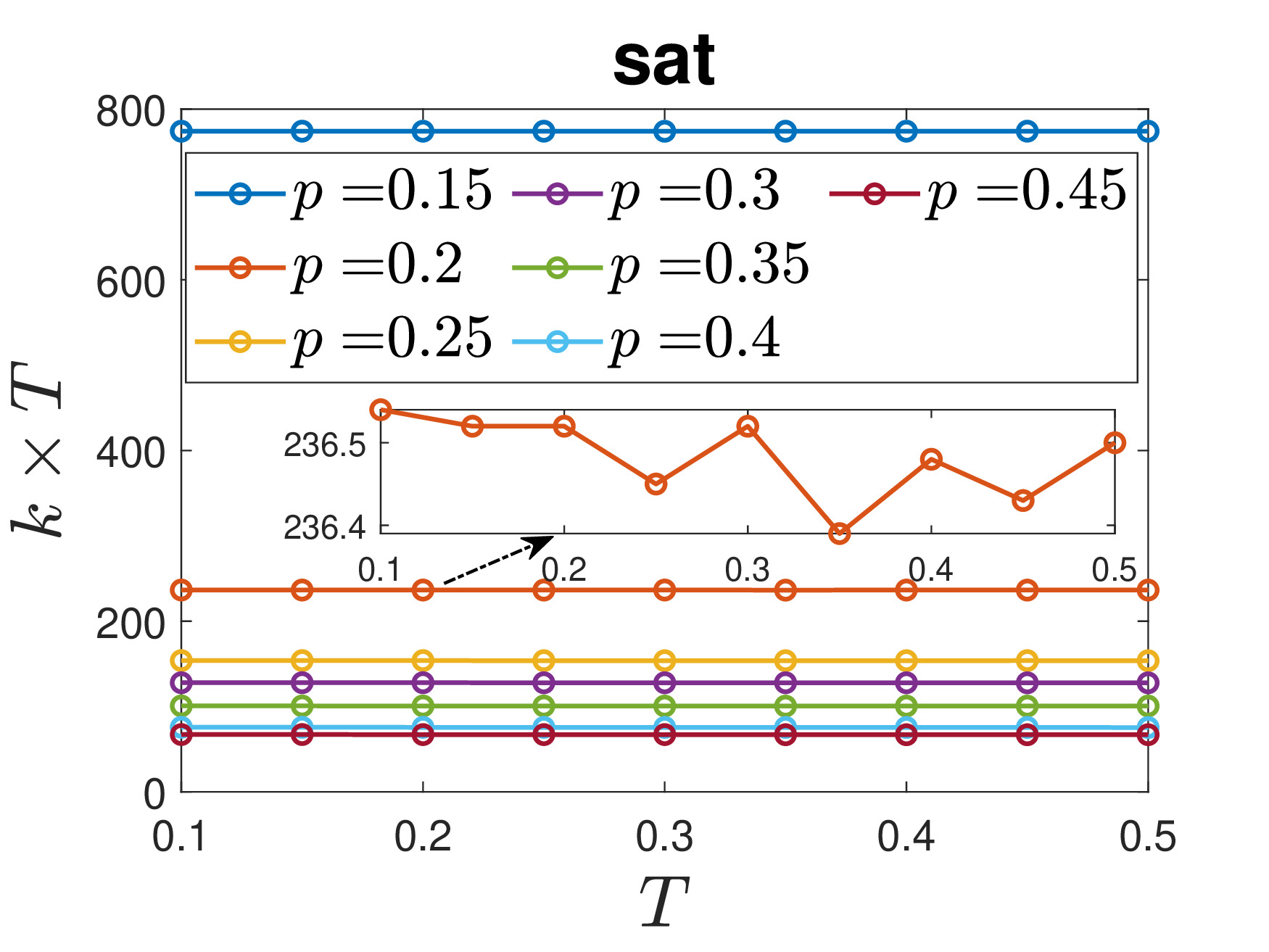}
	\includegraphics[width=1.6in]{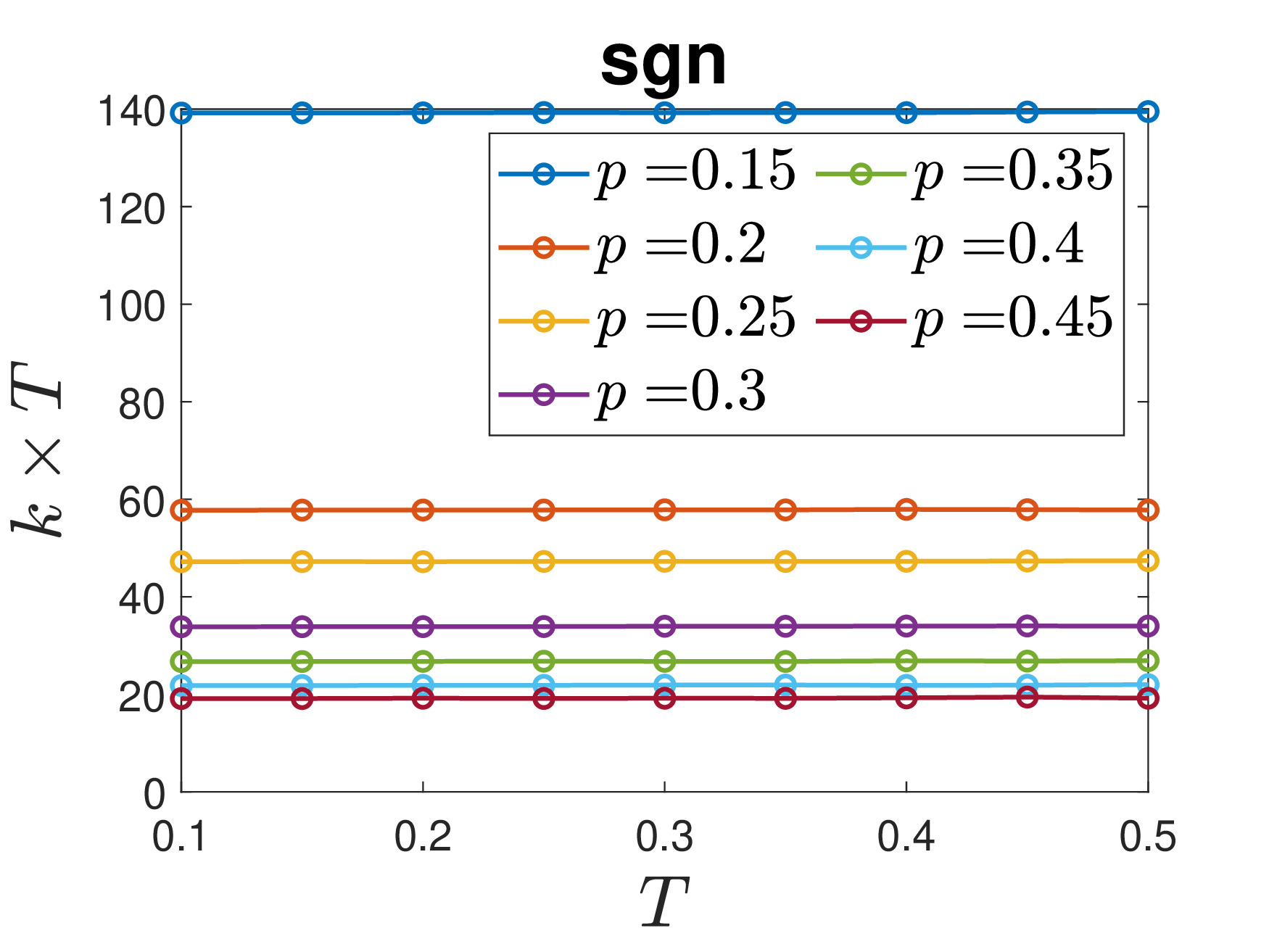}
	\caption{The convergence iteration$\times T$ to reach residual $\overline{F} = 10^{-2}$  versus step size $T$ for sat (level $\kappa=3$) and sgn  ($\nu_1=0.5$, $\nu_2=1.5$) nonlinearity. The connectivity of the ER network is changed via the link density $p$. }
	\label{fig_resid_p_T}
\end{figure}

\subsection{Allocation with Delayed Information-Exchange}
Next, we consider a cyclic network of $n=5$ generators under RRL limits  $\kappa=\frac{1}{60}$ $MW/sec$ in the presence of time delays. The data in Table~\ref{tab_par} is used which resembles IEEE 14-Bus test system \cite{Kar6345156}. The time-evolution of the generated powers and the Lyapunov function $\overline{F}$ (the residual) are simulated for random parameters.
For the quadratic cost \eqref{eq_f_quad}, $u=\max\{\gamma_i\} = 0.04$. We have $\varepsilon = 0.0166$, $K_g = 1$, $\lambda_2 = 1.38$, $T=1$ and $\lambda_n = 3.61$.
Using Eq.~\eqref{eq_Trange}, for any $T<0.045$ the solution converges in the absence of delays. This is a sufficient bound and to some extent not very tight.  

\textbf{Case I:}
In this case, using \eqref{eq_sol_efm_delay2_case}, we update the generator states at every $\overline{\tau}$ steps as in Section~\ref{sec_caseI}.  Following Remark~\ref{rem_delay_1slot}, we consider information exchange over the longer time-scale $\overline{k}$. 
The time-evolution of the Lyapunov function $\overline{F}$ (representing the residual) for some values of $\overline{\tau}$ is given in Fig.~\ref{fig_edp_caseI}. Clearly, the convergence rate decreases with the increase in the time delays.

\begin{figure}[]
	\centering
	\includegraphics[width=3.2in]{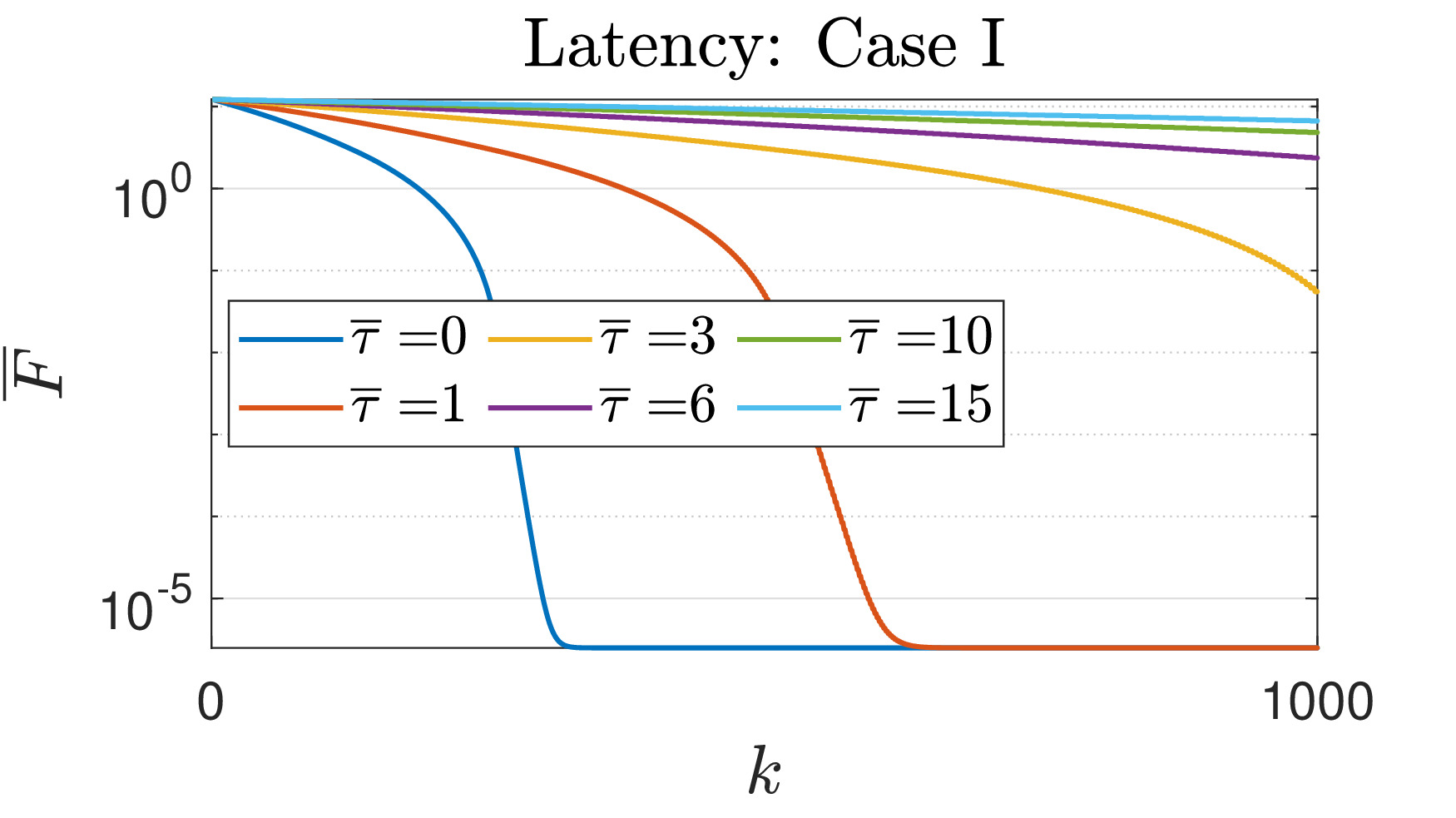}
	\caption{The residual $\overline{F}$ under protocol~\eqref{eq_sol_efm_delay2_case} (Case I) subject to  saturation and time-delays. This method can handle large max delays $\overline{\tau}$, even though the convergence rate is low. }
	\label{fig_edp_caseI}
\end{figure}

\textbf{Case II:}   
In this case, we update the generator states at every step based on all the received (delayed and non-delayed) information using \eqref{eq_sol_efm_delay2} in Section~\ref{sec_caseIII}. 
From Theorem~\ref{thm_converg_delay}, \textit{sufficient} condition for convergence is $T(\overline{\tau}+1) < 0.045$. For this simulation, however, we see that solution converges  $\overline{\tau} \leq 3$. We perform the simulation for both heterogeneous time-varying and time-invariant (fixed) delays as shown in Fig.~\ref{fig_edp_caseII}. 
\begin{figure}[]
	\centering
	\includegraphics[width=3.2in]{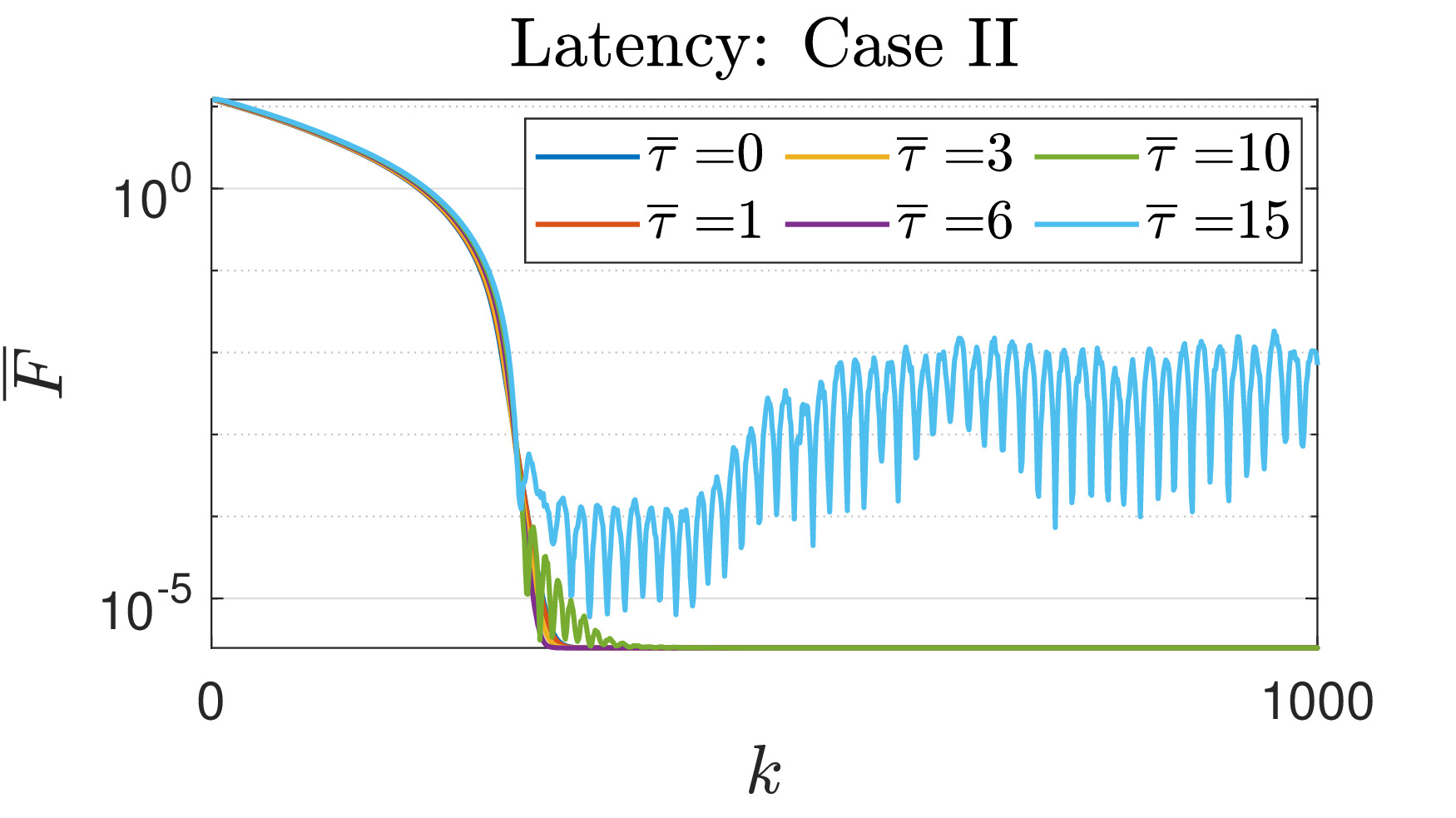}
	\includegraphics[width=3.2in]{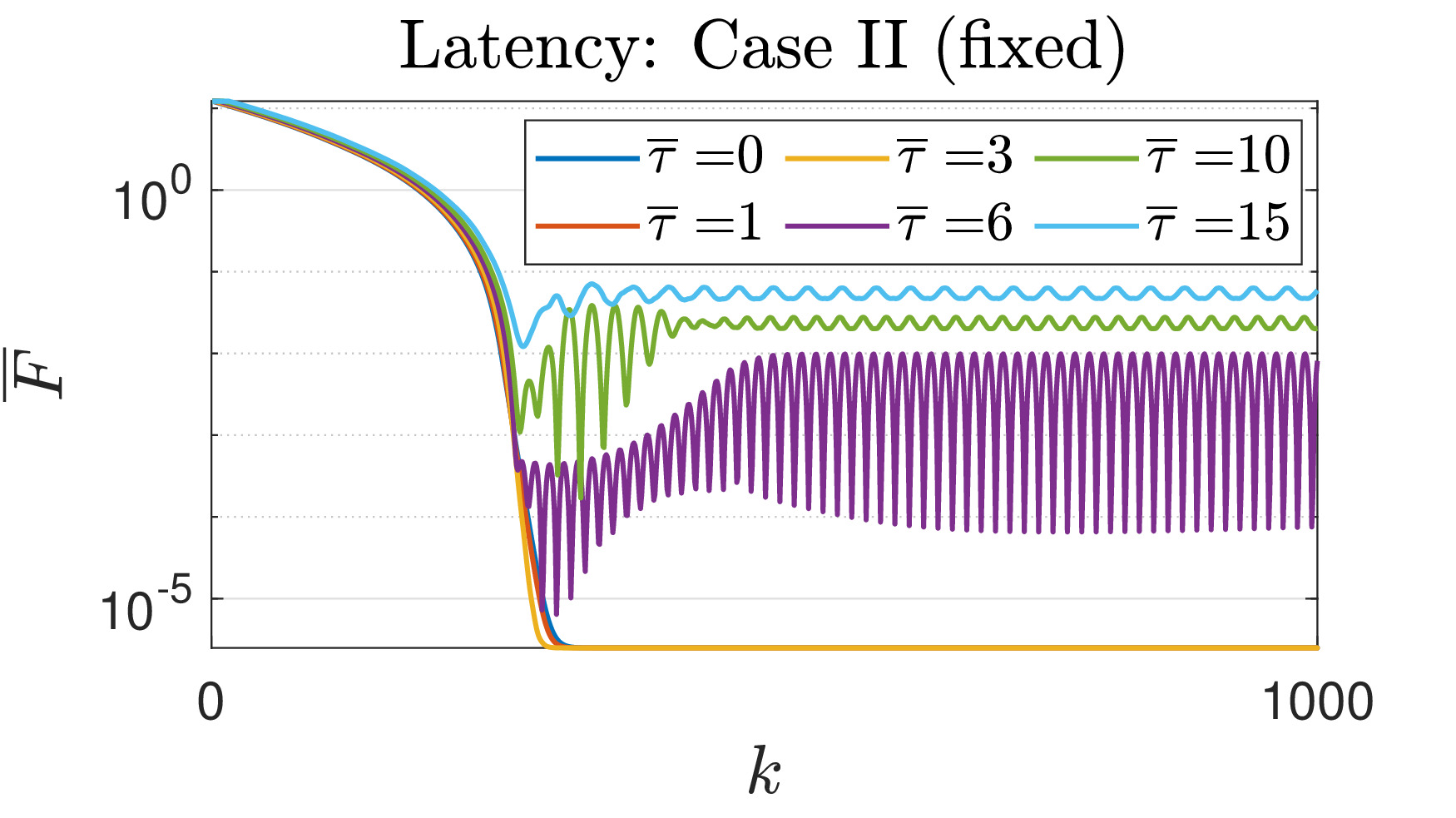}
	\caption{The residual $\overline{F}$ under DT protocol~\eqref{eq_sol_efm_delay2} under Case II with (top) time-varying and (bottom) time-invariant delays. For small $\overline{\tau}$ the convergence is faster than Case I, but for large $\overline{\tau}$  this approach may not converge (e.g., $\overline{\tau} = 15$).}
	\label{fig_edp_caseII}
\end{figure}

\section{Simulation of Non-smooth Quantized Optimization over Uniformly Connected Networks} \label{sec_sim}

In this section, we run the simulations over a random  \textit{dynamic} network of $n=100$ agents. The network is not connected at any time, while it is B-connected over every $B=100$ iterations (Assumption~\ref{ass_G}), i.e.,  $\bigcup_{k}^{k+100}\mc{G}(k)$ is connected. We consider strongly-convex logarithmic functions $f_i(\cdot)$ defined as \cite{boyd2006optimal},
\begin{align} \label{eq_F2}
	f_i(x_i) = \frac{1}{2}\alpha_i(x_i-\gamma_i)^2 + \zeta \log(1+\exp(\beta_i(x_i-\eta_i))),
\end{align}
with  the  coefficients randomly chosen as $0<\alpha_i<0.2$, $-0.2<\beta_i<0.2$, $-0.3<\gamma_i<0.3$, $0<\eta_i<0.6$, $\zeta = 0.2$, and random constraint parameters in problem \eqref{eq_dra0} chosen as $-2<a_i<2$ and $b=10$. To optimize this objective we apply~\eqref{eq_sol_efm2} (with $T=0.1$) under logarithmic quantization $g_{l}(\cdot)$ in Eq. \eqref{eq_quan_log}. This locally Lipschitz function satisfies $\Big(1-\frac{\delta}{2}\Big)z \leq \frac{g_{l}(z)}{z} \leq \Big(1+\frac{\delta}{2}\Big)z $ and, thus, Assumption~\ref{ass_gW} holds. We considered a composition of dynamics \eqref{eq_sol_efm}-\eqref{eq_sol_efm2} via $g_\kappa\Big(g_l(\partial f_i) -  g_l(\partial f_j)\Big)$ with $\kappa = 1$. The time-evolution of the residual $\overline{F}$ for different quantization levels $\delta$ is given in Fig.~\ref{fig_log_quant}. We further compare the performance under different time-delay models by assigning random heterogeneous delays to the links.
For Case II, two scenarios are given (i) time-varying, and (ii) time-invariant (fixed) delays. For Case I, from Remark~\ref{rem_delay_1slot}, we consider updating over time-scale $\overline{k}$ using composition of dynamics~\eqref{eq_sol_efm_delay_case}-\eqref{eq_sol_efm_delay2_case}. 
Simulations are shown in Fig.~\ref{fig_log_quant_delay1}
with parameters: $\delta = 0.125$, $T=0.05$, $\overline{\tau} = 2,6$.
As shown in Fig.~\ref{fig_log_quant_delay1}, the solution by Case II 
does not necessarily converge while the solution by Case I converges. This simulation shows that for small $\overline{\tau}$ (satisfying Theorem~\ref{thm_converg_delay}) Case II leads to faster convergence. On the other hand, for larger $\overline{\tau}$ Case I is a better delay-tolerant mechanism.

Note that the notion of information quantization and time-delays although prevalent in real-world networked systems cannot be addressed by the existing primal-based \cite{boyd2006optimal,gharesifard2013distributed,shames2011accelerated,doan2017scl,doan2017ccta,nedic2018improved,yi2016initialization,yang2013consensus,wang2018distributed,cherukuri2015distributed,lakshmanan2008decentralized} and dual-based solutions \cite{banjac2019decentralized,falsone2020tracking,chang2016proximal,wei_me_cdc,falsone2018distributed,aybat2019distributed,dtac}. In other words, these existing literature assume ideal network condition and there is no guarantee that their solution converge under quantized information and/or data transmission delays. 

\begin{figure}[]
	\centering
	\includegraphics[width=3.3in]{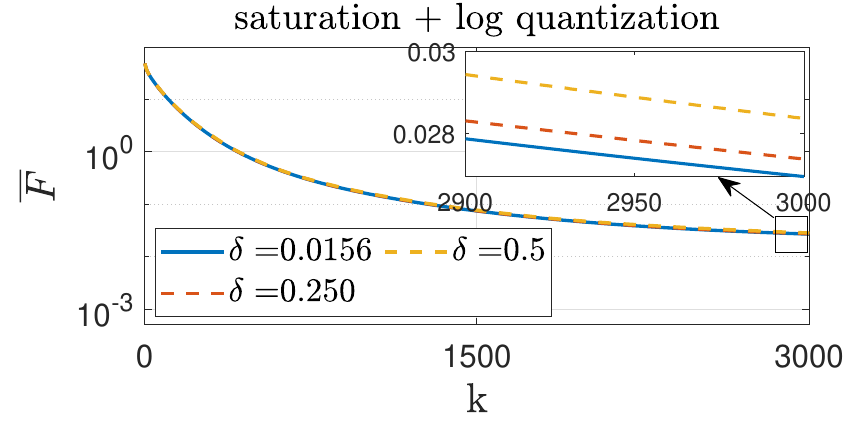}
	\caption{This figure shows the residual  $\overline{F}$ for quantized communications and saturated actuation (composition of nonlinear protocols~\eqref{eq_sol_efm}-\eqref{eq_sol_efm2}) over a dynamic uniformly-connected network. $\delta$ denotes the quantization level.}
	\label{fig_log_quant}
\end{figure}

\begin{figure}[]
	\centering
	\includegraphics[width=3.2in]{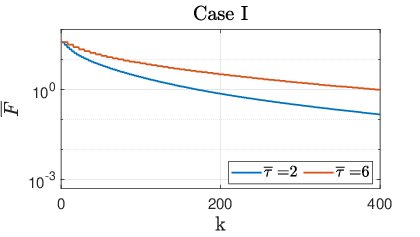}
	\includegraphics[width=3.2in]{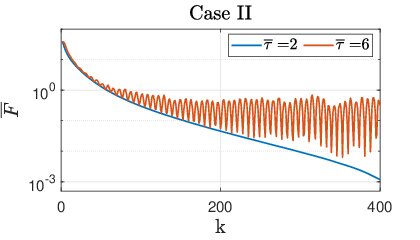}
	\includegraphics[width=3.2in]{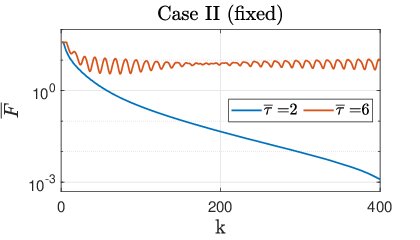}
	\caption{This figure shows the time-evolution of the residual $\overline{F}$ subject to delays $\overline{\tau} = 2,6$ and logarithmic quantization.  For Case II both time-varying and time-invariant (fixed at every link) delays are considered. For the case of a smaller delay value, the solution by Case II converges faster, while, for the larger delay value $\overline{\tau} = 6$, the solution by Case I gives better performance. }
	\label{fig_log_quant_delay1}
\end{figure}

\section{Conclusion and Future Directions} \label{sec_conclusion}
This paper proposes anytime-feasible (Laplacian-gradient) solutions subject to model nonlinearities to solve distributed sum-preserving resource allocation and coupling-constraint optimization over uniformly-connected networks (not necessarily connected at all times). The convergence to the optimal value is proved for general strongly sign-preserving nonlinearities. In addition, two scenarios are proposed to overcome heterogeneous delays over the network. For large delays and low buffers, we proposed to update the agents' states over a longer time scale after receiving (at least) $1$ delayed packet over every link. On the other hand, faster convergence can be achieved, for smaller time delays, by updating at the same time scale of the communication and using all the received (possibly) delayed packets at each iteration. The results are given for undirected and balanced graphs with bounded and link-symmetric delays.   

Allocation strategies over lossy networks with link failure or packet drop are another direction of our current research. As future research directions, application to asynchronous scheduling under quantized dynamics \cite{themis2021cpu}, robust (and noise-resilient) sign-based \cite{wei2017consensus,stankovic2019robust}, or single-bit \cite{taes2020finite} dynamics are of interest. Recall that for non ``strongly'' sign-preserving solutions, the $\ve$-accuracy needs to be addressed to give an estimate of the optimality gap. Another  future research direction is to extend the results to non-convex problems, which is  a bottleneck and more interesting for the industry.   

\bibliographystyle{elsarticle-num} 
\bibliography{bibliography}

\end{document}